\documentclass[12pt]{article}
\usepackage{amsmath, amssymb}
\usepackage[margin=1.25in]{geometry}
\usepackage{hyperref}
\usepackage{bbm}
\usepackage{xcolor}
\usepackage[authoryear]{natbib}
\usepackage{setspace}

\newtheorem{theorem}{Theorem}
\newtheorem{proposition}{Proposition}
\newtheorem{lemma}{Lemma}
\newtheorem{corollary}{Corollary}
\newtheorem{assumption}{Assumption}
\newtheorem{definition}{Definition}
\newtheorem{remark}{Remark}
\DeclareMathOperator{\rank}{rank}
\DeclareMathOperator{\spn}{span}
\DeclareMathOperator{\diag}{diag}

\newenvironment{proof}[1][Proof]{\noindent \textbf{#1.} }{\  \rule{0.5em}{0.5em}}

\begin{document}

\title{Identification and Estimation in Panel Data Models with Nonparametric
Two-Way Fixed Effects}
\author{Hugo Freeman\thanks{
Department of Economics, Michigan State University, East Lansing, USA.
E-mail: freem391@msu.edu} \and Dennis Kristensen\thanks{
Department of Economics, University College London, London, UK. E-mail:
d.kristensen@ucl.ac.uk} }
\date{\today}
\maketitle

\begin{abstract}
We study identification and estimation of panel data models that are characterised by a common, potentially, infinite-dimensional parameter, $\theta$ and an unknown transformation $g(\alpha _{i},\gamma _{t})$ of two-way
fixed effects $(\alpha _{i},\gamma _{t})$ that are themselves unknown. We treat $g$ as a nonparametric object to be identified/estimated together with $\theta$ and $(\alpha _{i},\gamma _{t})$. We  first provide high-level assumptions under which $\theta$ together with the image of $g$, $\{g(\alpha _{i},\gamma _{t})\}$, is identified. Next, we show that, under weak regularity conditons on $g$ and the underlying distribution of $(\alpha _{i},\gamma _{t})$, this image in turn allows us to separately identify normalised versions of g and the fixed effects. This is achieved by using a finite number of leading singular functions of $g$ as proxies (normalised versions) of the fixed effects. Our identification argument is constructive in the sense that it leads to natural estimators of the parameters of interest.
\end{abstract}


\section{Introduction}


We develop novel identification results for a class of panel data models
with nonparametric two-way unobserved heterogeneity and a potentially
nonparametric common component. We take as starting point that a given model
in this class can be characterised by an objective function $q\left(
Z_{i,t};\theta ,g\left( \alpha _{i},\gamma _{t}\right) \right) $, where $%
Z_{i,t}$ is observed, $\theta \in \Theta $ is an unknown common parameter
that is potentially infinite dimensional, e.g., $\theta \left(
X_{i,t}\right) $, where $X_{i,t}\in Z_{i,t}$, $g$ is an unknown function,
and $\left( \alpha _{i},\gamma _{t}\right) $ are unobserved. We impose no
restrictions on the dependence between $\left( \alpha _{i},\gamma
_{t}\right) $ and $Z_{i,t} $ and treat the realised values of the former as
two-way fixed effects to be estimated together with $\theta $ and $g$. In
order to identify these fixed effects together with $\theta $ and $g$, we
consider data $\left\{ Z_{i,t}:i=1,...,N,t=1,...,T\right\} $ with $%
N,T\rightarrow \infty $. An object that changes with $(N,T)$, such as the
set $\{g(\alpha _{i},\gamma _{t})\}_{i,t}$ or the proxies below, is called
identified when it is determined, as $N,T\rightarrow \infty $, by the
distribution of the data.

One example that fits into this framework is%
\begin{equation}
Y_{i,t}=\theta \left( X_{i,t}\right) +g\left( \alpha _{i},\gamma _{t}\right)
+\varepsilon _{i,t},\text{ \ \ }\mathbb{E}\left[ \varepsilon
_{it}|X_{i,t},\alpha _{i},\gamma _{t}\right] =0,
\label{eq: regression model}
\end{equation}%
where $Z_{i,t}=\left( Y_{i,t},X_{i,t}\right) $ is observed. Least-squares is
a natural estimator of $(\theta ,g,\left\{ \alpha _{i},\gamma _{t}\right\}
_{i,t})$ leading to the objective function $q\left( Z_{i,t};\theta ,g\left(
\alpha _{i},\gamma _{t}\right) \right) =\left( Y_{i,t}-\theta \left(
X_{i,t}\right) -g\left( \alpha _{i},\gamma _{t}\right) \right) ^{2}$.
Estimators of a semiparametric version of this model, with $\theta \left(
X_{i,t}\right) =\theta ^{\prime }X_{i,t}$, have been developed in %
\citet{freeman2023linear,beyhum2025inference,freeman2026inference}.

Our theory relies on the following fundamental assumption that we impose on $%
g$: We require it to lie in a Hilbert space $L_{\pi }^{2}(\Omega _{\alpha
}\times \Omega _{\gamma })$, where $\pi $ is the joint density of $\left(
\alpha _{i},\gamma _{t}\right) $ and $\Omega _{\alpha }$ and $\Omega
_{\gamma }$ are the supports of the $\alpha _{i}$ and $\gamma _{t}$,
respectively. This implies that $g$ can be represented as%
\begin{equation}
g(\alpha ,\gamma )=\sum_{r=1}^{\infty }\sigma _{r}u_{r}(\alpha )v_{r}(\gamma
),  \label{eq: g rep}
\end{equation}%
where $\left\{ u_{r}\right\} _{r=1}^{\infty }$ and $\left\{ v_{r}\right\}
_{r=1}^{\infty }$ are the singular functions and $\left\{ \sigma
_{r}\right\} $ the singular values. Our general identification and
estimation strategy then proceeds in two steps:

In Section \ref{sect:firststep}, we develop sieve--based extremum estimators
of $\theta $ and the set $\left\{ g(\alpha _{i},\gamma _{t})\right\} _{i,t}$
based on $q\left( z;\theta ,g\left( \alpha ,\gamma \right) \right) $ and
show that under high--level conditions these are consistent. As such, this
section delivers a constructive identification result of these two
components. The first-stage estimation employs the representation (\ref{eq:
g rep}). Specifically, we introduce the following sieve approximation of $g$%
, $g_{R}(\alpha _{i},\gamma _{t})=\sum_{r\leq R}\sigma _{r}u_{r}(\alpha
_{i})v_{r}(\gamma _{t})$, where the truncation parameter $%
R=R_{N,T}\rightarrow \infty $ with sample size. Writing this as $%
g_{R}(\alpha _{i},\gamma _{t})=\{\lambda _{i}^{\prime }f_{t}\}$, we show how
it is possible to estimate $\lambda _{i}$ and $f_{t}$, and thereby $g_{R}$,
using existing results from the literature on interactive fixed--effects.

In Section \ref{sect:secondstep}, we then present a general theory for
separate identification of normalised versions of $g$ and $\left\{ \alpha
_{i},\gamma _{t}\right\} _{i,t}$ from the identified (estimated) image $%
\left\{ g(\alpha _{i},\gamma _{t})\right\} _{i,t}$. This relies on the
high-level assumption of valid proxies $\left( \lambda _{i},f_{t}\right) $
being available: We require that there exists injective mappings $U\left(
\alpha \right) $ and $V\left( \gamma \right) $, so that $\lambda
_{i}:=U\left( \alpha _{i}\right) $ and $f_{t}:=V\left( \gamma _{t}\right) $
are identified. Note here that $U$ and $V$ will in general not be
identified; rather, they correspond to the particular normalisations
enforced on $\alpha _{i}$ and $\gamma _{t}$. Knowledge of $\lambda _{i}$ and 
$f_{t}$ in turn allows us to identify $h\left( \lambda ,f\right)
:=g(U^{-1}(\lambda ),V^{-1}(f))$. We then present a general estimator of $h$
based on estimated versions of $\lambda _{i}$ and $f_{t}$. Finally, we
propose a second--stage estimator of $\theta $ that utilises the fact that
we can now separately estimate $h$, $\lambda _{i}$ and $f_{t}$. This
second-stage estimator will in general dominate the preliminary estimator of 
$\theta $ developed in Section \ref{sect:firststep} since we are now able to
enforce local smoothness conditions on $h$ in estimation; see %
\citet{freeman2026inference} for details on this in the case of \ref{eq:
regression model} with $\theta \left( x\right) =\theta ^{\prime }x$.

Finally, in Sections \ref{sect:setup} and \ref{sect:mainResult}, we show
that the leading eigenfunctions of the representation (\ref{eq: g rep}) are
valid proxies -- that is, they satisfy our high-level conditions stated in
Section \ref{sect:secondstep} -- under primitive conditions on $g$ and $\pi $%
. We have already shown in Section \ref{sect:firststep} that these are
identified. What remains is then to show that finitely many leading singular
functions are injective w.r.t. $\alpha _{i}$ and $\gamma _{t}$,
respectively. This is done by first arguing that the response map $\alpha
\mapsto g(\alpha ,\cdot )$ is injective by construction through an
observational equivalence reduction. Next, we show that a local immersion
condition carries this injectivity over to a finite truncation.

All proofs have been relegated to the appendix.

\section{First-step identification and estimation}

\label{sect:firststep}

Let $\theta _{0},g_{0}$, $\alpha _{0,i}$ and $\gamma _{0,t}$ denote the data
generating values. We here show that $\theta _{0}$ and $\Gamma
_{0,N,T}=\left\{ g_{0,i,t}\right\} :=\left\{
g_{0,i,t}:i=1,...,N,t=1,...,T\right\} \in \mathbb{R}^{N\times T}$, where $%
g_{0,i,t}=g_{0}\left( \alpha _{0,i},\gamma _{0,t}\right) $, are identified
under weak regularity conditions on $q$ and $\left\{ Z_{i,t}\right\} $. The
proofs of this section can be found in Appendix~\ref{app:proofs-front}.

We assume that $\theta _{0}\in (\Theta ,\left\Vert \cdot \right\Vert )$.
Matrices in $\mathbb{R}^{N\times T}$ are measured in the scaled Frobenius
norm%
\begin{equation*}
\left\Vert \Gamma \right\Vert _{NT}:=(NT)^{-1/2}\left\Vert \Gamma
\right\Vert _{F},
\end{equation*}%
so that $\left\Vert \Gamma _{0,N,T}\right\Vert _{NT}^{2}$ is the empirical
mean of $g_{0}^{2}$ over the realised fixed effects. Throughout, the
supports $\Omega _{\alpha }$ and $\Omega _{\gamma }$ of the fixed effects
are compact and $g_{0}$ is continuous on $\Omega _{\alpha }\times \Omega
_{\gamma }$, as Assumption~\ref{ass:domain} below imposes, so with $B$ any
known constant satisfying $B\geq 2\sup |g_{0}|$ we have $\left\Vert \Gamma
_{0,N,T}\right\Vert _{NT}\leq B/2$ for every $N,T$. Let $R_{N,T}\rightarrow
\infty $ be the sieve rank, with $R_{N,T}<\min \{N,T\}$, and define the sieve%
\begin{equation*}
\mathcal{G}_{NT}=\left\{ \Gamma =\Lambda F^{\top }:\ \Lambda \in \mathbb{R}%
^{N\times R_{N,T}},\ F\in \mathbb{R}^{T\times R_{N,T}},\ \left\Vert \Gamma
\right\Vert _{NT}\leq B\right\} ,
\end{equation*}%
containing the matrices of rank at most $R_{N,T}$. Write $g_{R}$ for the
truncated function,%
\begin{equation}
g_{R}(\alpha ,\gamma )=\sum_{r=1}^{R}\sigma _{r}u_{r}(\alpha )v_{r}(\gamma ),
\label{eq: truncated kernel}
\end{equation}%
whose values at the realised fixed effects form the matrix $\{\lambda
_{i}^{\prime }\Sigma _{R}f_{t}\}$, which has rank at most $R$ and norm at
most $\sup |g_{0}|+o(1)$ under Assumption~\ref{ass:S} below, so it lies in $%
\mathcal{G}_{NT}$ for $R\leq R_{N,T}$ and all large $N,T$. The sieve does
not contain $\Gamma _{0,N,T}$ itself, whose rank is generically $\min
\{N,T\} $. It is compact, and it is not convex. The best approximation of $%
\Gamma _{0,N,T}$ in $\mathcal{G}_{NT}$ is its rank-$R_{N,T}$ truncated
singular value decomposition, by the theorem of Eckart and Young, which is
unique whenever the $R_{N,T}$-th and $(R_{N,T}+1)$-th singular values of $%
\Gamma _{0,N,T}$ differ, its factors being unique up to $(\Lambda ,F)\mapsto
(\Lambda H,FH^{-\top })$ for invertible $H$. The squared approximation error
in $\left\Vert \cdot \right\Vert _{NT}$ is the tail $%
\sum_{k>R_{N,T}}s_{k}^{2}/NT $, with $s_{k}$ the singular values of $\Gamma
_{0,N,T}$. It vanishes as $R_{N,T}\rightarrow \infty $ (Appendix~\ref%
{app:proofs-front}).

If $\theta _{0}$ is infinite dimensional, let $\Theta _{N,T}\subseteq \Theta 
$ be a compact sieve for it, and let $\Theta _{N,T}=\Theta $ otherwise, with 
$\Theta $ compact. We equip the product space $\Theta _{N,T}\times \mathcal{G%
}_{NT}$ with the norm $\left\Vert \left( \theta ,\Gamma \right) \right\Vert
:=\bigl(\Vert \theta \Vert ^{2}+\Vert \Gamma \Vert _{NT}^{2}\bigr)^{1/2}$.

For a given $\left( \theta _{N,T},\Gamma _{N,T}\right) \in \Theta
_{N,T}\times \mathcal{G}_{NT}$, define the sample objective function.%
\begin{equation*}
\hat{Q}_{N,T}\left( \theta _{N,T},\Gamma _{N,T}\right) =\frac{1}{NT}%
\sum_{i=1}^{N}\sum_{t=1}^{T}q\left( Z_{i,t};\theta _{N,T},g_{i,t}\right) .
\end{equation*}%
We then estimate the unknown parameters by any approximate minimiser,%
\begin{equation*}
\hat{Q}_{N,T}(\hat{\theta}_{N,T},\hat{\Gamma}_{N,T})\leq \inf_{\left( \theta
,\Gamma \right) \in \Theta _{N,T}\times \mathcal{G}_{NT}}\hat{Q}_{N,T}\left(
\theta ,\Gamma \right) +o_{P}(1).
\end{equation*}%
The sample objective may not have a well-defined limit as $N,T\rightarrow
\infty $. We therefore introduce%
\begin{equation*}
Q_{N,T}\left( \theta _{N,T},\Gamma _{N,T}\right) =\frac{1}{NT}%
\sum_{i=1}^{N}\sum_{t=1}^{T}\mathbb{E}\left[ q\left( Z_{i,t};\theta
_{NT},g_{i,t}\right) |\mathcal{F}_{N,T}\right] ,
\end{equation*}%
where $\mathcal{F}_{N,T}$ is the sigma algebra generated by the
data-generating fixed effects $\{\alpha _{0,i}\}_{i\leq N}$ and $\{\gamma
_{0,t}\}_{t\leq T}$ and by the regressors $\{X_{i,t}\}$, so that $Q_{N,T}$
treats both as fixed, as in \citet{Bai2009} and \citet{freeman2023linear}.
We impose the following high-level conditions, where we use the following
notation to save space: $\vartheta =\left( \theta ,\Gamma \right) $, $%
\vartheta _{0,N,T}=\left( \theta _{0},\Gamma _{0,N,T}\right) $, and $%
\mathcal{A}_{N,T}=\Theta _{N,T}\times \mathcal{G}_{NT}$:

\begin{assumption}
\label{ass:A1} (i) $\sup_{\vartheta \in \mathcal{A}_{N,T}}|\hat{Q}%
_{N,T}\left( \vartheta \right) -Q_{N,T}\left( \vartheta \right)
|=o_{P}\left( 1\right) $; (ii) $Q_{N,T}$ is continuous on $\mathcal{A}_{NT}$
and, with $\vartheta _{N,T}^{\ast }=(\theta _{N,T}^{\ast },\Gamma
_{N,T}^{\ast })$ a minimiser of $Q_{N,T}$ over $\mathcal{A}_{NT}$, for every 
$\epsilon >0$%
\begin{equation*}
\liminf_{N,T\rightarrow \infty }\ \inf_{\vartheta \in \mathcal{A}%
_{NT}:\left\Vert \vartheta -\vartheta _{N,T}^{\ast }\right\Vert \geq
\epsilon }\Bigl\{Q_{N,T}(\vartheta )-Q_{N,T}(\vartheta _{N,T}^{\ast })\Bigr\}%
>0\text{ a.s.;}
\end{equation*}%
(iii) the minimiser converges to the truth, $\left\Vert \vartheta
_{N,T}^{\ast }-\vartheta _{0,N,T}\right\Vert \rightarrow 0$.
\end{assumption}

The minimiser in (ii) exists by the following lemma, and we call it the
pseudo-true value. Because $Q_{N,T}$ conditions on $\mathcal{F}_{N,T}$, the
pseudo-true value is random, and the separation in (ii) and the convergence
in (iii) are required to hold with probability approaching one, that is, for
every $\epsilon >0$ there is an $\eta >0$, not depending on the realisation,
such that the infimum in (ii) is at least $\eta $ and $\Vert \vartheta
_{N,T}^{\ast }-\vartheta _{0,N,T}\Vert <\epsilon $ with probability
approaching one. The second half of (ii) is identifiable uniqueness uniform
in the sample. It is what replaces the pointwise uniqueness of a fixed limit
objective, which is not available here because both the parameter space and $%
Q_{N,T}$ change with $(N,T)$.

Part (iii) is needed because the sieve may not contain the truth. When $%
g_{0} $ has infinite rank, $\Gamma _{0,N,T}$ lies outside $\mathcal{G}_{NT}$%
, and the pseudo-true value $\vartheta _{N,T}^{\ast }$ is the closest point
of the sieve to the truth, in the sense of minimising $Q_{N,T}$. Parts (i)
and (ii) make the estimator converge to this pseudo-true value, and part
(iii) asks that the pseudo-true value in turn converge to the truth as $%
R_{N,T}$ grows. In the regression example below this holds because the
approximation error is the tail of the singular values of $g_{0}$ beyond $%
R_{N,T}$.

\begin{lemma}
\label{lem:pseudotrue} The set $\Theta _{N,T}\times \mathcal{G}_{NT}$ is
compact, and so $Q_{N,T}$ attains its minimum on it under Assumption~\ref%
{ass:A1}(ii).
\end{lemma}

\begin{theorem}
\label{thm:firststep} Under Assumption~\ref{ass:A1}, $\left\Vert \hat{\theta}%
_{N,T}-\theta _{0}\right\Vert +\left\Vert \hat{\Gamma}_{N,T}-\Gamma
_{0,N,T}\right\Vert _{NT}\rightarrow _{p}0$.
\end{theorem}

Assumption~\ref{ass:A1}(i) is a uniform law of large numbers and holds under
standard regularity conditions. The sieve $\mathcal{G}_{NT}$ is parametrised
by the $R_{N,T}(N+T)$ entries of $\Lambda $ and $F$, which are estimated
from $NT$ observations, so the number of parameters per observation, $%
R_{N,T}(N+T)/NT$, vanishes only if $R_{N,T}$ grows more slowly than $\min
\{N,T\}$.

Parts (ii) and (iii) carry the identification content. In the regression
model (\ref{eq: regression model}) we condition on all regressors and fixed
effects, which requires strict exogeneity, $\mathbb{E}[\varepsilon _{i,t}|%
\mathcal{F}_{N,T}]=0$. Then%
\begin{equation*}
Q_{N,T}\left( \theta ,\Gamma \right) -Q_{N,T}\left( \theta _{0},\Gamma
_{0,N,T}\right) =\frac{1}{NT}\sum_{i=1}^{N}\sum_{t=1}^{T}\mathbb{E}\left[
\left( \theta _{0}\left( X_{i,t}\right) -\theta \left( X_{i,t}\right)
+g_{0,i,t}-g_{i,t}\right) ^{2}|\mathcal{F}_{N,T}\right] ,
\end{equation*}%
a quadratic centred at the truth. The rank restriction on $\Gamma $ allows
for the matrix $\{\theta _{0}(X_{i,t})-\theta (X_{i,t})\}$ to be separately
identified from $\Gamma $. What identifies $\theta _{0}$ is variation in $%
X_{i,t}$ that is not a function of the fixed effects, $(\alpha _{i},\gamma
_{t})$. For $\theta $, part (ii) then asks that, for every $\epsilon >0$,
there is an $\eta >0$, not depending on $(N,T)$, such that with probability
approaching one no $\theta \in \Theta _{N,T}$ with $\Vert \theta -\theta
_{N,T}^{\ast }\Vert \geq \epsilon $ has $\{\theta (X_{i,t})-\theta
_{N,T}^{\ast }(X_{i,t})\}$ within $\Vert \cdot \Vert _{NT}$-distance $\eta $
of the matrices of rank at most $2R_{N,T}$, the largest rank of a difference
of two sieve elements. This is a generalised version of the no perfect
collinearity condition of \citet{freeman2023linear}, who consider the linear
case $\theta \left( X_{i,t}\right) =\theta ^{\prime }X_{i,t}$. Part (iii) is
the content of the consistency proof of \citet{freeman2023linear}, where $%
R_{N,T}\rightarrow \infty $ at a controlled rate and the approximation error
is the tail $\sum_{r>R_{N,T}}\sigma _{r}^{2}$ of the singular values.

One could define $Q_{N,T}$ by conditioning only on the fixed effects, but
not the regressors. This would impose the weaker contemporaneous exogeneity
of $\mathbb{E}[\varepsilon _{i,t}|X_{i,t},\mathcal{F}_{N,T}^{\circ }]=0$ In the regression
model (\ref{eq: regression model}),
where $\mathcal{F}_{N,T}^{\circ }$ is the sigma algebra generated by the
fixed effects. The variation in $X_{i,t}$ that is not a function of the
fixed effects would then identify $\theta _{0}$ without the rank
restriction, which would be needed only for estimation. We here follow
condition \citet{Bai2009} and \citet{freeman2023linear} and also on the
regressors, so that their results apply in the linear case.

Any $\theta $ whose matrix $\{\theta (X_{i,t})\}$ has bounded rank, such as
a constant or a function of a regressor that varies only over $i$ or only
over $t$, is absorbed by $\Gamma $, so $\Theta _{N,T}$ must exclude these
directions. This is what the non-collinearity condition does in the linear
case.

This establishes a first-step estimate of the identified set $\Gamma
_{0,N,T}=\left\{ g_{0}(\alpha _{0,i},\gamma _{0,t})\right\} $. The
estimators of $\theta _{0}$ and $\Gamma _{0,N,T}$ will generally not attain
an optimal rate, since the sieve does not impose the regularity properties
that the underlying $g_{0}$ and the distributions of $(\alpha _{0,i},\gamma
_{0,t})$ may have. The first step fits $R_{N,T}(N+T)$ unrestricted entries,
carries the truncation bias, and uses no smoothness of $g_{0}$. As a
consequence, $\hat{\theta}_{N,T}$ will generally not be optimal.

We therefore now proceed to the second step to our identification and
estimation procedure, where we use $\Gamma _{0,N,T}$ to identify $g_{0}$ and 
$\left\{ \alpha _{0,i},\gamma _{0,t}\right\} $ up to some normalisation.
These then allow us to develop a second-step estimator, wehere we can impose
any desired regularity conditions on $g_{0}$ and $\left\{ \alpha
_{0,i},\gamma _{0,t}\right\} $. The second step replaces the $R_{N,T}(N+T)$
entries with a smooth function of finitely many proxies for the fixed
effects, which is what allows \citet{freeman2026inference} to show the $%
\sqrt{NT}$ rate for $\theta _{0}$ in the semiparametric regression model. It
uses only the leading $R$ singular vectors of $\hat{\Gamma}_{N,T}$, with $R$
fixed in the uniform regime and $R(\delta )$ in the almost-everywhere regime
of Theorem~\ref{thm:headline}, so the growth of the sieve rank $R_{N,T}$
serves the identification of $\theta _{0}$ and $\Gamma _{0,N,T}$, not the
construction of the proxies.

\section{Second-step identification and estimation}

\label{sect:secondstep} 

In this section, we first provide general conditions under which knowledge
of $\Gamma _{0,N,T}$ allows us to identify separately $g_{0}$ and $(\alpha
_{0,i},\gamma _{0,t})$ up to some normalisation. That is, we identify $%
\lambda _{0,i}\in \mathbb{R}^{d_{\lambda }}$ and $f_{0,t}\in \mathbb{R}%
^{d_{f}}$, for some finite $d_{\lambda },d_{f}$, and some $h_{0}$ so that%
\begin{equation}
g_{0}(\alpha _{0,i},\gamma _{0,t})=h_{0}\left( \lambda _{0,i},f_{0,t}\right)
.  \label{eq: ID cond}
\end{equation}%
For any bijections $\phi $ and $\psi $ of $\Omega _{\alpha }$ and $\Omega
_{\gamma }$ onto their images, the function $g_{0}\circ (\phi ^{-1},\psi
^{-1})$ evaluated at $(\phi (\alpha _{0,i}),\psi (\gamma _{0,t}))$
reproduces the same identified set, so $g_{0}$ and the fixed effects are
identified only up to such a reparametrisation. The proxies of Definition~%
\ref{def:proxyID} below fix it by taking $\phi =U$ and $\psi =V$, and $h_{0}$
is $g_{0}$ written in these coordinates. The proof of this section can be
found in Appendix~\ref{app:proofs-front}.

This identification result is again constructive in that it leads to natural
estimators of $\lambda _{0,i}$, $f_{0,t}$ and $h_{0}$. These estimators can
in turn be used to obtain a second--step estimator of $\theta $ that we
expect will enjoy better properties compared to the initial estimator $\hat{%
\theta}_{N,T}$ of the previous section.

In the following, we suppress the subscript null on all data-generating
parameters and write $g$ and $(\alpha _{i},\gamma _{t})$ for notational
ease, with $\Omega _{\alpha }\subseteq \mathbb{R}^{d_{\alpha }}$ and $\Omega
_{\gamma }\subseteq \mathbb{R}^{d_{\gamma }}$ the supports of the two fixed
effects. The first step was written for scalar $g$ for notational
simplicity. The techniques and the results can easily be generalised to
allow for models whose objective function takes the form $q\left(
Z_{i,t};\theta ,g_{1}\left( \alpha _{i},\gamma _{t}\right)
,...,g_{d_{G}}\left( \alpha _{i},\gamma _{t}\right) \right) $. From here on $%
g=\left( g_{1},...,g_{d_{G}}\right) $ is $\mathbb{R}^{d_{G}}$ valued.
Vector-valued $g$ covers models in which several equations or components of $%
q$ share the fixed effects, and more components make the fixed effects
easier to separate, since the full-rank condition of Assumption~\ref%
{ass:multi-index} uses $d_{G}p_{\alpha }\geq d_{\alpha }$ rows, so that
fewer evaluation points $p_{\alpha }$ can suffice as $d_{G}$ grows. The
first step extends by stacking the $d_{G}$ matrices $\Gamma ^{(j)}=\Lambda
_{j}F^{\top }$ with a common right factor $F$, which is the structure the
singular value decomposition of Section~\ref{sect:setup} imposes.

Our strategy for separate identification of $g$ and $\left( \alpha
_{i},\gamma _{t}\right) $ relies on the following notion of valid proxies of
the unknown fixed effects:

\begin{definition}
\label{def:proxyID} We call $\lambda _{i}=U(\alpha _{i})\in \mathbb{R}%
^{d_{\lambda }}$ and $f_{t}=V(\gamma _{t})\in \mathbb{R}^{d_{f}}$, with $U,V$
Borel measurable and $d_{\lambda },d_{f}<\infty $, valid proxies if:

\begin{enumerate}
\item[(i)] $\lambda _{i}$ and $f_{t}$ are identified as $N,T\rightarrow
\infty $.

\item[(ii)] $U$ is injective on $\Omega _{\alpha }$ and $V$ is injective on $%
\Omega _{\gamma }$.
\end{enumerate}
\end{definition}

\begin{theorem}
\label{th:proxyID} Let $\lambda _{i}=U(\alpha _{i})$ and $f_{t}=V(\gamma
_{t})$ be valid proxies in the sense of Definition~\ref{def:proxyID}, let
the set $\{g(\alpha _{i},\gamma _{t})\}$ be identified, and let $g$ be Borel
measurable. Then 
\begin{equation*}
h(\lambda ,f):=g\bigl(U^{-1}(\lambda ),V^{-1}(f)\bigr),\qquad (\lambda
,f)\in U(\Omega _{\alpha })\times V(\Omega _{\gamma }),
\end{equation*}%
is Borel measurable, satisfies (\ref{eq: ID cond}) at every realised pair,
and is identified at those pairs. If moreover $g$, $U$ and $V$ are
continuous and Assumption~\ref{ass:S} holds, then $h$ is identified on all
of $U(\Omega _{\alpha })\times V(\Omega _{\gamma })$. If in addition $g$ is $%
C^{q}$ for some $q\geq 1$, and $U$ and $V$ are injective $C^{q}$ immersions,
then $h$ is $C^{q}$, with differentiability on the image manifolds $%
U(\Omega _{\alpha })$ and $V(\Omega _{\gamma })$.
\end{theorem}

\noindent The first claim uses injectivity only at the realised fixed
effects, so it also holds for the realised pairs with fixed effects in
subsets of the supports on which $U$ and $V$ are injective. This is how the
theorem is used in the almost-everywhere regime of Theorem~\ref{thm:headline}%
.

Theorem~\ref{th:proxyID} identifies $h$ through its values at the realised
pairs. Once the identified set and the proxies are replaced by first-step
estimates, any local smoother of the identified set on the proxies estimates 
$h$. For instance, kernel regression could be employed,%
\begin{equation*}
\hat{h}(\lambda ,f)=\frac{\sum_{i=1}^{N}\sum_{t=1}^{T}\hat{g}%
_{i,t}K_{b_{\lambda }}(\hat{\lambda}_{i}-\lambda )\,K_{b_{f}}(\hat{f}_{t}-f)%
}{\sum_{i=1}^{N}\sum_{t=1}^{T}K_{b_{\lambda }}(\hat{\lambda}_{i}-\lambda
)\,K_{b_{f}}(\hat{f}_{t}-f)},
\end{equation*}%
with kernels $K$ and bandwidths $b_{\lambda },b_{f}\rightarrow 0$. Its
population version, with $\hat{g}_{i,t}$, $\hat{\lambda}_{i}$ and $\hat{f}%
_{t}$ replaced by $g(\alpha _{i},\gamma _{t})$, $U(\alpha _{i})$ and $%
V(\gamma _{t})$, is a weighted average of $g$ over the realised pairs whose
proxies lie within the bandwidths of $(\lambda ,f)$, and converges to $%
h(\lambda ,f)$ under Assumption~\ref{ass:S} whenever $h$ is continuous at $%
(\lambda ,f)$. For $\hat{h}$ itself, the proxy estimation error must vanish
faster than the bandwidths, so that the estimated proxies fall in the same
kernel windows as the true ones. Importantly, by suitable choice of kernels
and bandwidths, above estimator can adapt to any smoothness conditions on $g$
that $h$ inherits,

It could now be tempting to use above kernel regression estimator to obtain
an updated version of $\hat{\theta}_{N,T}$ by solving $\min_{\theta _{N,T}\in \Theta _{N,T}}\frac{1}{NT}\sum_{i=1}^{N}%
\sum_{t=1}^{T}q( Z_{i,t};\theta _{N,T},\hat{h}(\hat{\lambda}_{i},\hat{f}%
_{t}))$.
However, the resulting estimator of $\theta _{N,T}$ will generally not be
Neyman orthogonal w.r.t. $\hat{h}$ and so will suffer from larger than
necessary finite-sample errors due to $\hat{h}$. Instead, we propose to
estimate $\theta _{N,T}$ and $h$ jointly using either kernel or sieve
methods.

The kernel version proceeds in two steps: For a given candidate value of $%
\theta $, we first profile out $h$ by%
\begin{equation*}
\hat{h}\left( \lambda ,f;\theta \right) =\arg \min_{h}\frac{%
\sum_{i=1}^{N}\sum_{t=1}^{T}q\left( Z_{i,t};\theta ,h\right) K_{b_{\lambda
}}(\hat{\lambda}_{i}-\lambda )\,K_{b_{f}}(\hat{f}_{t}-f)}{%
\sum_{i=1}^{N}\sum_{t=1}^{T}K_{b_{\lambda }}(\hat{\lambda}_{i}-\lambda
)\,K_{b_{f}}(\hat{f}_{t}-f)}.
\end{equation*}%
This is in turn used to estimate $\theta $ by%
\begin{equation*}
\hat{\theta}_{N,T}^{\left( 2\right) }=\arg \min_{\theta _{N,T}\in \Theta
_{N,T}}\frac{1}{NT}\sum_{i=1}^{N}\sum_{t=1}^{T}q\left( Z_{i,t};\theta _{N,T},%
\hat{h}(\hat{\lambda}_{i},\hat{f}_{t};\theta )\right) .
\end{equation*}%
This can be seen as a generalised version of the estimator of %
\citet{freeman2026inference}. The sieve--based estimation procedure first
defines a suitable sieve space $\mathcal{H}_{NT}$ for the function $h$ and
then computes 
\begin{equation*}
(\tilde{\theta}_{N,T}^{\left( 2\right) },\tilde{h})=\arg \min_{\substack{ %
\theta _{N,T}\in \Theta _{N,T} \\ h\in \mathcal{H}_{NT}}}\frac{1}{NT}%
\sum_{i=1}^{N}\sum_{t=1}^{T}q(Z_{i,t};\theta _{N,T},h(\hat{\lambda}_{i},\hat{%
f}_{t})).
\end{equation*}%
Unlike the first step, which fits $R_{N,T}(N+T)$ unrestricted entries,
carries the truncation bias and uses no smoothness of $g$, these two second
step estimators solve a lower--dimensional problem, where smoothness
conditions can be imposed on $h$. The second-step estimators will, however,
still suffer from incidental parameter biases in general due to the presence
of $\{\hat{\lambda}_{i},\hat{f}_{t}\}$ in the procedure. Such incidental
parameter biases can be removed by orthogonalising the estimator of $\theta
_{N,T}$ w.r.t. these. For further details on this, we refer to %
\citet{freeman2026inference} who use such techniques to develop an estimator
of $\theta _{N,T}$ that is $\sqrt{NT}$-asymptotically normally distributed .

\begin{remark}[Dimension of the second-step regression]
\label{rem:kernel} One could be concerned that above estimators would suffer
from an overly severe curse-of-dimensionality if $d_{\lambda }>d_{\alpha }$
or $d_{f}>d_{\gamma }$ since this would seem to make the dimension of the
domain of $h$ very large. However, this will generally not be the case if
the injective maps $U$ and $V$ have Jacobians of full rank $d_{\alpha }$ and 
$d_{\gamma }$, respectively, uniformly over the supports. In this case, the
proxies lie on the image of $\Omega _{\alpha }\times \Omega _{\gamma }$
under a $C^{1}$ embedding, a submanifold of dimension $d_{\alpha }+d_{\gamma
}$ regardless of the values of $d_{\lambda }$ and $d_{f}$. Kernel--based
estimators in the ambient proxy coordinates therefore attains the rate of a $%
(d_{\alpha }+d_{\gamma })$-dimensional regression \citep{bickel2007local}.
No reduction of the proxies to $\mathbb{R}^{d_{\alpha }}$ and $\mathbb{R}%
^{d_{\gamma }}$ is therefore needed, although the constants grow with $R$.
For sieve estimators to enjoy similar attractive properties, the sieve space
has to be chosen carefully; one such is Neural Networks.
\end{remark}

The remainder of the paper develops conditions under which a finite number
of the leading singular functions of $g$, collected in $U_{R}=(u_{1}^{\top
},\dots ,u_{R}^{\top })^{\top }$ and $V_{R}=(v_{1},\dots ,v_{R})^{\top }$,
are injective, and hence valid proxies once identified.

This particular choice of proxies leads to the following version of above
two-step estimatin procedure: Given $\hat{\Gamma}_{N,T}=\left\{ \hat{g}%
_{i,t}\right\} \in \mathcal{G}_{NT}$ from the first step, take its singular
value decomposition $\hat{\Gamma}_{N,T}=\sum_{k}\hat{s}_{k}\hat{u}_{k}\hat{v}%
_{k}^{\top }$ and, for a fixed $R\geq R^{\ast }$, set $\hat{\lambda}_{i}:=%
\sqrt{N}(\hat{u}_{1,i},\dots ,\hat{u}_{R,i})^{\prime }$ and $\hat{f}_{t}:=%
\sqrt{T}(\hat{v}_{1,t},\dots ,\hat{v}_{R,t})^{\prime }$, each determined up
to sign, which is immaterial since a sign flip of a proxy coordinate is an
invertible linear map that preserves injectivity and smoothness. Injectivity
needs $R\geq R^{\ast }$, which the singular values do not reveal, and in the
simulations of \citet{freeman2026inference} the optimal hard threshold of %
\citet{gavish2014optimal} is used as a practical choice. Next, compute $\hat{%
\theta}_{N,T}^{\left( 2\right) }$ or $\tilde{\theta}_{N,T}^{\left( 2\right)
} $ based on these proxy estimates.

\section{Identification of singular functions of \texorpdfstring{$g$}{g}}

\label{sect:setup}

This section formally defines the singular functions of $g$, and shows that
they are identified from $\Gamma _{0,N,T}$. Throughout, $\Omega _{\alpha
}\subseteq \mathbb{R}^{d_{\alpha }}$ and $\Omega _{\gamma }\subseteq \mathbb{%
R}^{d_{\gamma }}$ are the supports of the two fixed effects, and the
probability measures $\pi _{\alpha }$ on $\Omega _{\alpha }$ and $\pi
_{\gamma }$ on $\Omega _{\gamma } $ have densities with respect to Lebesgue
measure bounded away from $0$ and $\infty $. Notation $L^{2}$ means scalar
valued, $\mathcal{H}$ means $\mathbb{R}^{d_{G}}$ valued, and the subscript
names the support. The subscript $\pi $ always denotes the measure of the
indicated support, i.e. $\pi _{\alpha }$ on $\Omega _{\alpha }$, $\pi
_{\gamma }$ on $\Omega _{\gamma }$, and $\pi _{\alpha }\otimes \pi _{\gamma
} $ on $\Omega _{\alpha }\times \Omega _{\gamma }$. The scalar spaces then
need no further symbol. For the vector-valued spaces we abbreviate the $%
L_{\pi }^{2}$ mappings to $\mathbb{R}^{d_{G}}$ as 
\begin{equation*}
\mathcal{H}_{\gamma }:=L_{\pi }^{2}(\Omega _{\gamma };\mathbb{R}%
^{d_{G}}),\qquad \mathcal{H}_{\alpha }:=L_{\pi }^{2}(\Omega _{\alpha };%
\mathbb{R}^{d_{G}}),
\end{equation*}%
with inner products $\langle \phi ,\,\psi \rangle _{L_{\pi }^{2}(\Omega
_{\gamma })}=\int_{\Omega _{\gamma }}\phi \psi \,d\pi _{\gamma }$ and $%
\langle \phi ,\,\psi \rangle _{\mathcal{H}_{\gamma
}}=\sum_{j=1}^{d_{G}}\langle \phi _{j},\,\psi _{j}\rangle _{L_{\pi
}^{2}(\Omega _{\gamma })}$, the $\alpha $-side case being analogous. These
are the only spaces used in the main text, and each hosts one kind of
object. The scalar right singular functions $v_{r}$ below live in $L_{\pi
}^{2}(\Omega _{\gamma })$, the $\mathbb{R}^{d_{G}}$-valued left singular
functions $u_{r}$ and the slices $g(\cdot ,\gamma )$ live in $\mathcal{H}%
_{\alpha }$, and the responses $g(\alpha ,\cdot )$ live in $\mathcal{H}%
_{\gamma }$. Mean-square convergence on the product support is convergence
in $L_{\pi }^{2}(\Omega _{\alpha }\times \Omega _{\gamma };\mathbb{R}%
^{d_{G}})$. We write $\lVert \cdot \rVert $ for the Euclidean norm on $%
\mathbb{R}^{k}$, for the Frobenius norm of a matrix, and for the operator
norm of an operator, the meaning fixed by the argument.

\subsection{The integral operator and its SVD}

Let $g\in L^{2}_{\pi}(\Omega_{\alpha}\times\Omega_{\gamma};\mathbb{R}^{d_G})$%
, that is, $\int_{\Omega_{\alpha}}\int_{\Omega_{\gamma}}\lVert
g\rVert^2\,d\pi_\gamma\,d\pi_\alpha<\infty$. Define the Hilbert-Schmidt
integral operator $T_g$ 
\begin{align}
T_g:L^{2}_{\pi}(\Omega_{\gamma})\longrightarrow\mathcal{H}_{\alpha}, \qquad
(T_g\phi)(\alpha)=\int_{\Omega_{\gamma}}
g(\alpha,\gamma)\,\phi(\gamma)\,d\pi_\gamma(\gamma) \ \in\mathbb{R}^{d_G},
\label{eq:Tg}
\end{align}
and its associated dual $T_g^*$ 
\begin{align*}
T_g^*:\mathcal{H}_{\alpha}\longrightarrow\,L^{2}_{\pi}(\Omega_{\gamma}),
\qquad (T_g^*\psi)(\gamma)=\int_{\Omega_{\alpha}}
g(\alpha,\gamma)^{\top}\psi(\alpha)\,d\pi_\alpha(\alpha).
\end{align*}
By the singular value decomposition of a compact operator between Hilbert
spaces, there exist a nonincreasing sequence of singular values $%
\sigma_1\ge\sigma_2\ge\cdots\geq0$, an orthonormal system $%
\{u_r\}_{r\ge1}\subset\mathcal{H}_{\alpha}$ spanning $\overline{\text{range}%
(T_g)}$, and an orthonormal system $\{v_r\}_{r\ge1}\subset\,L^{2}_{\pi}(%
\Omega_{\gamma})$ spanning $\overline{\text{range}(T_g^*)}$, such that by 
\cite{reed1980methods}, 
\begin{equation}
g(\alpha,\gamma)=\sum_{r = 1}^\infty\sigma_r\,u_r(\alpha)\,v_r(\gamma),
\label{eq:svd}
\end{equation}
with convergence in $L^{2}_{\pi}(\Omega_{\alpha}\times\Omega_{\gamma};%
\mathbb{R}^{d_G})$. Functions $u_r/v_r$ are the \emph{left/right singular
functions}; the left ones are $\mathbb{R}^{d_G}$ valued and the right ones
scalar.

The asymmetry in \eqref{eq:svd} reflects the choice of operator. The domain
of $T_g$ is scalar and its codomain is $\mathbb{R}^{d_G}$ valued, and each
family of singular functions lives in this space. Coordinatewise, %
\eqref{eq:svd} says $g_j(\alpha,\gamma)=\sum_r\sigma_r
u_{r,j}(\alpha)v_r(\gamma)$ for $j=1,\dots,d_G$. At each rank $r$ a \emph{%
single} scalar $v_r$ is shared by all $d_G$ components of $g$, and the
component-specific information sits in the vector $u_r(\alpha)$. The $v_r $
are eigenfunctions of $T_g^{*}T_g$, whose kernel $\int_{\Omega_{\alpha}}g(%
\alpha,\gamma)^{\top}g(\alpha,\gamma^{\prime })\,d\pi_\alpha$ pools
variation across components, so in general $v_r$ is not a singular function
of any single component $g_j$. The pooling is what the proxy needs, and a
pair of time fixed effects separated by any one component is then separated
by the common system. Appendix~\ref{sec:gamma} derives $\gamma$-side proxies
through the adjoint, within this single system.

\subsection{Singular function maps and truncations}

For $R\in\mathbb{N}\cup\{\infty\}$ define the (possibly infinite) \emph{left
singular function map} 
\begin{equation}
U_R:\Omega_{\alpha}\longrightarrow\mathbb{R}^{d_G R}, \qquad U_R(\alpha):=%
\bigl(u_1(\alpha)^{\top},\dots,u_R(\alpha)^{\top}\bigr)^{\top},
\label{eq:UR}
\end{equation}
with the convention $U_\infty(\alpha)=\{u_r(\alpha)\}_{r=1}^\infty$. We call 
$U_R$ for finite $R$ the \emph{leading truncation} of order $R$, since by %
\eqref{eq:svd} the components are ordered by decreasing singular value. The 
\emph{response map} is 
\begin{equation}
\Phi:\Omega_{\alpha}\longrightarrow\mathcal{H}_{\gamma}, \qquad
\Phi(\alpha):=g(\alpha,\cdot).  \label{eq:Phi}
\end{equation}
The truncated function is $g_R(\alpha,\gamma):=\sum_{r=1}^R\sigma_r
u_r(\alpha) v_r(\gamma)$ and the truncated response map is $%
\Phi_R(\alpha):=g_R(\alpha,\cdot)$.

\subsection{Identification of the proxies}

\label{sec:proxies-identified}

The singular value decomposition, and every result from here on, is taken
with respect to $\pi _{\alpha }$ and $\pi _{\gamma }$, whereas the first
step of Section~\ref{sect:firststep} concerns realised fixed effects. One
assumption links the two.

\begin{assumption}[Fixed effects and the reference measures]
\label{ass:S} The empirical measures $N^{-1}\sum_{i\leq N}\delta _{\alpha
_{i}}$ and $T^{-1}\sum_{t\leq T}\delta _{\gamma _{t}}$ converge weakly to $%
\pi _{\alpha }$ and $\pi _{\gamma }$ as $N,T\rightarrow \infty $.
\end{assumption}

This assumption holds, for instance, when the fixed effects are independent
draws from $\pi _{\alpha }$ and $\pi _{\gamma }$, but no restriction on the
dependence between the fixed effects and $Z_{it}$ is imposed.

\begin{lemma}[Identified proxies]
\label{lem:proxies-identified} Suppose $\Gamma _{0,N,T}=\{g(\alpha
_{i},\gamma _{t})\}$ is identified, Assumption~\ref{ass:S} holds, $\Omega
_{\alpha }$ and $\Omega _{\gamma }$ are compact with $g$ continuous on $%
\Omega _{\alpha }\times \Omega _{\gamma }$ (Assumption~\ref{ass:domain}),
and $\sigma _{1}>\sigma _{2}>\cdots >\sigma _{R}>\sigma _{R+1} $. Let $%
\Gamma _{0,N,T}=\sum_{k}\tilde{s}_{k}\tilde{u}_{k}\tilde{v}_{k}^{\top }$ be
its singular value decomposition with $\tilde{s}_{1}\geq \tilde{s}_{2}\geq
\cdots $ and unit vectors $\tilde{u}_{k}\in \mathbb{R}^{N}$, $\tilde{v}%
_{k}\in \mathbb{R}^{T}$. Then, for every $r\leq R$, there are signs $%
\varsigma _{r}=\varsigma _{r,N,T}\in \{-1,1\}$ such that, as $N,T\rightarrow
\infty $, 
\begin{equation*}
\frac{\tilde{s}_{r}}{\sqrt{NT}}\rightarrow \sigma _{r},\qquad \max_{i\leq N}%
\bigl|\varsigma _{r}\sqrt{N}\,\tilde{u}_{r,i}-u_{r}(\alpha _{i})\bigr|%
\rightarrow 0,\qquad \max_{t\leq T}\bigl|\varsigma _{r}\sqrt{T}\,\tilde{v}%
_{r,t}-v_{r}(\gamma _{t})\bigr|\rightarrow 0.
\end{equation*}%
In particular $\lambda _{i}=U_{R}(\alpha _{i})$ and $f_{t}=V_{R}(\gamma
_{t}) $ are identified, up to the sign of each singular function. For the
regression model \eqref{eq: regression model} with $\theta (X)=\theta
^{\prime }X$ the set $\{g(\alpha _{i},\gamma _{t})\}$ is identified by %
\citet{freeman2023linear} under their rank condition on $X$, and this is
their singular function result.
\end{lemma}

\noindent The lemma is what makes the proxies usable. The rescaled leading
singular vectors of the identified set converge uniformly, up to sign, to
the singular functions at the realised fixed effects, so the proxies are
identified without knowledge of $g$ or of the fixed effects, and their
sample analogues are the estimated proxies of Section~\ref{sect:secondstep}.
The proof is in Appendix~\ref{app:proofs-setup}.

\section{Injectivity of the leading singular functions}

\label{sect:mainResult}

This section gives conditions under which the leading truncations $U_{R}$
and $V_{R}$ are injective, so that together with Lemma~\ref%
{lem:proxies-identified} they are valid proxies. The technical lemmas and
all proofs are in Appendix~\ref{app:injectivity}. Assumptions and results
are stated for $\alpha $ and $U_{R}$. The statements for $\gamma $ and $%
V_{R} $ exchange the roles of $\alpha $ and $\gamma $ but stay within the
singular system of \eqref{eq:svd}, through the adjoint rather than by
transposing $g$, so that one truncation order serves both sides. They are
given in Appendix~\ref{sec:gamma}.

An important object in our analysis is the following \emph{gradient
information matrix}, 
\begin{equation}
M_{R}(\alpha ):=\sum_{r=1}^{R}\sigma _{r}^{2}\,Du_{r}(\alpha )^{\top
}Du_{r}(\alpha )\;\in \;\mathbb{R}^{d_{\alpha }\times \,d_{\alpha }},
\label{eq:MR}
\end{equation}%
where $R\in \mathbb{N}\cup \{\infty \}$, $Du_{r}(\alpha )\in \mathbb{R}%
^{d_{G}\times \,d_{\alpha }}$ is the Jacobian of $u_{r}$ at $\alpha $, a row
vector transposed to a gradient when $d_{G}=1$. Each summand is positive
semidefinite, so $M_{R}(\alpha )\succeq 0$ and $R\mapsto M_{R}(\alpha )$ is
nondecreasing in the positive semidefinite order. The limit $M_{\infty
}(\alpha ):=\lim_{R\rightarrow \infty }M_{R}(\alpha )$ exists in the
positive-semidefinite order. Under Assumption~\ref{ass:domain} it is finite
and admits the closed form 
\begin{equation}
M_{\infty }(\alpha )=\int_{\Omega _{\gamma }}J(\alpha ,\gamma )^{\top
}J(\alpha ,\gamma )\,d\pi _{\gamma }(\gamma ),\qquad J(\alpha ,\gamma
):=D_{\alpha }g(\alpha ,\gamma )\in \mathbb{R}^{d_{G}\times \,d_{\alpha }};
\label{eq:Minf}
\end{equation}%
this, and the underlying gradient expansion $J(\alpha ,\cdot
)=\sum_{r}\sigma _{r}Du_{r}(\alpha )v_{r}$, converging entrywise in $L_{\pi
}^{2}(\Omega _{\gamma })$, is established in Lemma~\ref{lem:gradexp}. We use 
$J$ and $D_{\alpha }g$ interchangeably. Since $\sum_{r\leq R}\sigma
_{r}^{2}\lVert u_{r}(\alpha +\xi )-u_{r}(\alpha )\rVert ^{2}=\xi ^{\top
}M_{R}(\alpha )\xi +o(\lVert \xi \rVert ^{2})$, the matrix $M_{R}(\alpha )$
measures how fast the truncation $U_{R}$ separates two nearby fixed effects.
Its limit $M_{\infty }(\alpha )$ is positive definite exactly when no small
move of $\alpha $ leaves the response $g(\alpha ,\cdot )$ unchanged to first
order.


\subsection{Assumptions}


We collect sufficient conditions. Assumption~\ref{ass:domain} is standard,
and we show $u_r$ is $C^1$ on $\Omega_{\alpha}$ in Theorem~\ref%
{thm:primitive} under weak regularity. Assumption~\ref{ass:obseq} is the
observational-equivalence reduction. Assumption~\ref{ass:ae-immersion} is
used to show a finite set of singular functions inherits the injectivity
property pointwise almost-everywhere, and Assumption~\ref{ass:multi-index}
is its uniform strengthening. The price of the weaker Assumption~\ref%
{ass:ae-immersion} is that the result is demoted to pointwise
almost-everywhere.

\begin{assumption}[Supports and smoothness]
\label{ass:domain} Support $\Omega_{\alpha}\subseteq\mathbb{R}^{d_{\alpha}}$
is compact, convex, and connected, and $\Omega_{\gamma}\subseteq\mathbb{R}%
^{d_{\gamma}}$ is compact. The function $g$ is continuous on $%
\Omega_{\alpha}\times\Omega_{\gamma}$ and continuously differentiable in $%
\alpha$, with $D_\alpha g$ jointly continuous on $\Omega_{\alpha}\times%
\Omega_{\gamma}$; in particular each $u_r$ is $C^1$ on $\Omega_{\alpha}$.
\end{assumption}

\begin{assumption}[Observational equivalence]
\label{ass:obseq} Response map $\Phi$ in \eqref{eq:Phi} is injective by
construction. That is, for all $\alpha_1,\alpha_2\in\Omega_{\alpha}$, 
\begin{equation*}
g(\alpha_1,\gamma)=g(\alpha_2,\gamma)\ \text{ for all }\gamma\in\Omega_{%
\gamma} \quad\Longrightarrow\quad \alpha_1=\alpha_2.
\end{equation*}
\end{assumption}

\begin{remark}
\label{rem:obseq-justification} Observationally equivalent values induce
identical conditional responses and cannot be distinguished by any estimator
based on $g$. Passing to $\Omega_{\alpha}$ after removing observationally
equivalent support points is therefore without loss and renders $\Phi$
injective by construction; we write $\Omega_{\alpha}$ for the restricted
support.
\end{remark}

\begin{assumption}
\label{ass:ae-immersion} Degeneracy set $Z:=\{\alpha\in\Omega_{\alpha}:\det
M_\infty(\alpha)=0\}$ has $\pi_\alpha(Z)=0$.
\end{assumption}

\noindent Since $M_\infty(\alpha)\succeq0$ always, and a symmetric
positive-semidefinite matrix is positive definite if and only if its
determinant is nonzero, $\alpha\notin Z$ if and only if $M_\infty(\alpha)%
\succ0$: Assumption~\ref{ass:ae-immersion} is equivalent to requiring $%
M_\infty(\alpha)\succ0$ for $\pi_\alpha$-almost every $\alpha$. It permits a
rank drop on the $\pi_\alpha$-null set $Z$ (closed, by continuity of $%
M_\infty$, Lemma~\ref{lem:Minf-uniform} in Appendix~\ref{app:inj-basic}) and
is the condition of regime \textup{(i)} of Theorem~\ref{thm:headline}.
Assumption~\ref{ass:multi-index} next is its uniform strengthening. When $g$
is real analytic in $\alpha$ and $\pi_\alpha$ is absolutely continuous,
Assumption~\ref{ass:obseq} already implies Assumption~\ref{ass:ae-immersion}%
, so the analytic class carries the almost-everywhere regime for free (Lemma~%
\ref{lem:analytic} in Appendix~\ref{sect:analytic}).

\begin{assumption}
\label{ass:multi-index} There exists $b_0=(\gamma_{0,1},\dots,\gamma_{0,p_%
\alpha})\in\Omega_{\gamma}^{p_\alpha}$ with $d_{\alpha} \le d_G\,p_\alpha$
and $p_\alpha<\infty$ such that for every $\alpha\in\Omega_{\alpha}$ the
stacked Jacobian 
\begin{equation*}
D_\alpha G_1(\alpha,b_0) :=\Bigl[\,D_\alpha g(\alpha,\gamma_{0,1})\,;\
\cdots\ ;\, D_\alpha g(\alpha,\gamma_{0,p_\alpha})\,\Bigr] \in\mathbb{R}%
^{d_G p_\alpha\times\,d_{\alpha}}
\end{equation*}
has full column rank $d_{\alpha}$.
\end{assumption}

\begin{remark}
\label{rem:support} Because $\Omega_{\gamma}$ and $\Omega_{\alpha}$ are, by
the convention of Section~\ref{sect:setup}, the supports of $\pi_\gamma$ and 
$\pi_\alpha$, the evaluation points automatically lie in the supports; this
is what transfers the finite-design full rank to $M_\infty$ in %
\eqref{eq:Minf} (Lemma~\ref{lem:Minf-uniform}). The condition is local (an
immersion condition). Under Assumption~\ref{ass:domain} it is equivalent to $%
M_\infty(\alpha)\succ0$ at every $\alpha$, that is to $Z=\varnothing$, since
positive definiteness at $\alpha$ yields finitely many time fixed effects
whose stacked Jacobian has full rank there, full rank persists on a
neighbourhood, and stacking the designs of a finite subcover of $%
\Omega_{\alpha}$ gives one design that works everywhere. The two regimes
therefore differ only in whether the degeneracy set $Z$ is empty or merely $%
\pi_\alpha$-null. The symmetric $\gamma$-side condition, on a stacked
Jacobian $D_\gamma G_2(a_0,\gamma)$, is its $\gamma$-analogue, stated in
Appendix~\ref{sec:gamma} and used in Theorem~\ref{thm:headline}.
\end{remark}

\begin{remark}
\label{rem:independence} Assumptions~\ref{ass:obseq} and~\ref%
{ass:multi-index} are logically independent. With $d_{\alpha }=1$, $\Omega
_{\alpha }=[-1,1]$, and $g(\alpha ,\gamma )=\alpha ^{3}\ell (\gamma )$ (with 
$\ell :\Omega _{\gamma }\rightarrow \mathbb{R}^{d_{G}}$ nonconstant), $\Phi $
is injective (Assumption~\ref{ass:obseq} holds) yet $D_{\alpha }g(0,\gamma
)\equiv 0$, so Assumption~\ref{ass:multi-index} fails at $\alpha =0$.
Conversely, an everywhere full-rank Jacobian may self-intersect globally, so
Assumption~\ref{ass:multi-index} does not imply Assumption\ \ref{ass:obseq}.
\end{remark}


\subsection{Injectivity of the full collection}

The first result shows that injectivity of the full collection of singular
functions requires only Assumption~\ref{ass:obseq}. Parseval makes the
weighted coordinates $\sigma_r u_r(\alpha)$ an isometric copy of the
response $g(\alpha,\cdot)$, so the reduction of Assumption~\ref{ass:obseq}
costs nothing at the level of the infinite collection. The proofs are in
Appendix~\ref{app:inj-full}.

\begin{proposition}
\label{prop:part1} Let $g$ be square-integrable on $\Omega_{\alpha}\times%
\Omega_{\gamma}$ with the SVD \eqref{eq:svd}, where $r$ ranges over $%
\sigma_r>0$ throughout; let $\alpha\mapsto g(\alpha,\cdot)$ be continuous
into $\mathcal{H}_{\gamma}$ and take the continuous versions $%
u_r=\sigma_r^{-1}T_gv_r$ (both supplied by Assumption~\ref{ass:domain}).
Then for all $\alpha_1,\alpha_2\in\Omega_{\alpha}$, 
\begin{equation}
\lVert \Phi(\alpha_1)-\Phi(\alpha_2)\rVert_{\mathcal{H}_{\gamma}}^2
=\sum_{r\ge1}\sigma_r^2\,\lVert u_r(\alpha_1)-u_r(\alpha_2)\rVert^2.
\label{eq:parseval}
\end{equation}
Consequently the following are equivalent:

\begin{enumerate}
\item[\textup{(i)}] $\Phi$ is injective (Assumption~\ref{ass:obseq});

\item[\textup{(ii)}] the full singular function map $U_\infty$ is injective 
\textup{(}equal to $U_M$ when $\rank T_g=M<\infty$\textup{)}.
\end{enumerate}
\end{proposition}

\begin{corollary}
\label{cor:part1} Under Assumption~\ref{ass:obseq}, the map $U_\infty$ is
injective on $\Omega_{\alpha}$.
\end{corollary}

\subsection{Main result}

We state the headline result. It collects, in both regimes, the full set of
assumptions under which the truncated pair $(U_{R},V_{R})$ is injective, and
hence, together with Lemma~\ref{lem:proxies-identified}, supplies valid
proxies in the sense of Definition \ref{def:proxyID}. The $\gamma$-analogue
of an assumption is that assumption with the roles of $\alpha$ and $\gamma$
exchanged (Appendix~\ref{sec:gamma}).

\begin{theorem}[Injective singular-function proxies]
\label{thm:headline} Let Assumptions~\ref{ass:domain} and~\ref{ass:obseq}
hold together with their $\gamma$-analogues (Appendix~\ref{sec:gamma}), let
the set $\{g(\alpha_i,\gamma_t)\}$ be identified, let Assumption~\ref{ass:S}
hold, and let the leading singular values be distinct at the orders $R$ used
below, or else take $U_R$ and $V_R$ up to the block rotations in the proof
of Lemma~\ref{lem:proxies-identified} in Appendix~\ref{app:proofs-setup}.

\begin{enumerate}
\item[\textup{(i)}] \textup{(Almost-everywhere regime.)} If in addition
Assumption~\ref{ass:ae-immersion} and its $\gamma$-analogue hold, then for
every $\delta>0$ there exist a finite $R(\delta)$ and compact sets $%
\Omega_\delta^\alpha\subseteq\Omega_{\alpha}\setminus Z$ and $%
\Omega_\delta^\gamma\subseteq\Omega_{\gamma}\setminus Z_\gamma$ with $%
\pi_\alpha(\Omega_\delta^\alpha)\ge1-\delta$ and $\pi_\gamma(\Omega_\delta^%
\gamma)\ge1-\delta$, such that for every $R\ge R(\delta)$ the truncations $%
U_R$ and $V_R$ are injective on $\Omega_\delta^\alpha$ and $%
\Omega_\delta^\gamma$, respectively.

\item[\textup{(ii)}] \textup{(Uniform regime.)} If in addition Assumption~%
\ref{ass:multi-index} and its $\gamma$-analogue hold, then there is a finite 
$R^{*}$ such that for every $R\ge R^{*}$, $U_R$ is injective on $%
\Omega_{\alpha}$ and $V_R$ is injective on $\Omega_{\gamma}$.
\end{enumerate}

In either regime $(\lambda _{i},f_{t})=(U_{R}(\alpha _{i}),V_{R}(\gamma
_{t}))$ are identified, up to sign, by Lemma~\ref{lem:proxies-identified}.
Under \textup{(ii)} they are valid proxies in the sense of Definition~\ref%
{def:proxyID} at any fixed $R\geq R^{\ast }$. Under \textup{(i)} they
separate the fixed effects in $\Omega _{\delta }^{\alpha }\times \Omega
_{\delta }^{\gamma }$ with $\delta $ arbitrarily small, at the price of $%
R(\delta )\rightarrow \infty $ as $\delta \downarrow 0$.
\end{theorem}

The proof, in Appendix~\ref{app:injectivity}, has three steps. Under
Assumption~\ref{ass:obseq} the full collection of singular functions is
injective (Proposition~\ref{prop:part1}). Assumption~\ref{ass:obseq}
separates every pair of distinct fixed effects, but the truncated separation 
$\sum_{r\le R}\sigma_r^2\lVert u_r(\alpha_1)-u_r(\alpha_2)\rVert^2$
increases to the full one only pointwise, at an order that can diverge as
the two fixed effects approach each other. Compactness and Dini's theorem
therefore handle distant pairs with Assumption~\ref{ass:obseq} alone, while
nearby pairs need every direction of $\alpha$ to move the response at first
order, which Assumption~\ref{ass:ae-immersion} requires off a null set and
Assumption~\ref{ass:multi-index} everywhere (Proposition~\ref{prop:ae} and
Theorem~\ref{thm:main}). The $\gamma$ side follows by the same argument
through the adjoint.

\subsection{Choosing between the regimes}

\label{sec:regimes}

The two regimes of Theorem~\ref{thm:headline} trade the strength of the
maintained assumptions against the behaviour of the truncation order. For
identification of $g$ at the realised fixed effects (Theorem~\ref{th:proxyID}%
), either serves: under \textup{(i)} the identity $h(\lambda_i,f_t)=g(%
\alpha_i,\gamma_t)$ holds for fixed effects in $\Omega_\delta^\alpha\times%
\Omega_\delta^\gamma $ for every $\delta$, and so for almost every fixed
effect once $R$ is allowed to grow along $\delta\downarrow0$. For estimation
the comparison is quantitative, and it is governed by two distinct features
of the conflated set, its mass and its diameter. Fix $R$ and collect the
conflated fixed effects 
\begin{equation*}
B_R:=\bigl\{\alpha\in\Omega_{\alpha}:\ U_R(\alpha)=U_R(\alpha^{\prime })%
\text{ for some } \alpha^{\prime }\ne\alpha\bigr\}, \qquad
\delta_\alpha(R):=\pi_\alpha(B_R),
\end{equation*}
\begin{equation*}
\epsilon_\alpha(R):=\sup\bigl\{\lVert \alpha_1-\alpha_2\rVert:\
\alpha_1\neq\alpha_2,\ U_R(\alpha_1)=U_R(\alpha_2)\bigr\},
\end{equation*}
with $B_R^\gamma$, $\delta_\gamma(R)$, $\epsilon_\gamma(R)$ defined
symmetrically; each $B_R$ is Borel (a countable union of projections of
compact sets) and the sets are nested downward in $R$. Under Assumptions~\ref%
{ass:domain} and~\ref{ass:obseq} alone, $\epsilon_\alpha(R)\downarrow0$,
i.e. the away-from-the-diagonal step in the proof of Theorem~\ref{thm:main}
uses neither immersion condition, so for every $\epsilon>0$ every conflated
pair of separation at least $\epsilon$ is removed at some finite order.
Hence, for large $R$ any conflated pair is $\epsilon_\alpha(R)$-close.
Further, under the assumptions of regime \textup{(i)}, $\delta_\alpha(R)%
\downarrow0$, hence if $\alpha\notin Z$, its partners $\alpha^{\prime }_R$
must accumulate at $\alpha$ itself (any other limit point would be conflated
with $\alpha$ under $U_\infty$, contradicting Proposition~\ref{prop:part1}
and Assumption~\ref{ass:obseq}), so they eventually enter a compact
neighbourhood of $\alpha$ on which Proposition~\ref{prop:ae} makes some
finite truncation injective. Hence $\bigcap_R B_R\subseteq Z$ and continuity
from above gives $\pi_\alpha(B_R)\downarrow\pi_\alpha(\bigcap_R B_R)=0$.

A unit $i$ with $\alpha_i\in B_R$ potentially has $h(\lambda_i,f_t)\neq
g(\alpha_i,\gamma_t)$ at \emph{every} $t$, since below the injectivity
threshold $h$ is the regression of $g$ on $(U_R,V_R)$ under $%
\pi_\alpha\otimes\pi_\gamma$ and averages $g$ over the conflated fibre at
each $\gamma_t$, and symmetrically for a time $t$ with $\gamma_t\in
B_R^\gamma$. Hence, conflation contaminates whole rows and columns of an $%
N\times T$ array, so the fraction of affected cells is of order $%
\delta(R):=\delta_\alpha(R)+\delta_\gamma(R)$, not the product $%
\delta_\alpha(R)\,\delta_\gamma(R)$. The size of the error on an affected
cell is bounded by response variation across the conflated set. With $L$ the
joint Lipschitz constant of $g$, 
\begin{equation*}
\bigl|h(\lambda_i,f_t)-g(\alpha_i,\gamma_t)\bigr| \;\le\; L\,\epsilon(R),
\qquad \epsilon(R):=\epsilon_\alpha(R)+\epsilon_\gamma(R),
\end{equation*}
so the error shrinks with the spatial scale of the conflation (Remark~\ref%
{rem:blend-sharp}'s grid of size $2^{-R}$).

In procedures that orthogonalise with respect to the regression function,
first-order errors cancel and the conflation enters root-$NT$ normalised
sums only through second-order products of nuisance errors, both factors
carrying the bias on the same affected cells; the contribution is of order $%
\sqrt{NT}\,\delta(R)\,\bigl(L\epsilon(R)\bigr)^2$, while integrated
squared-error (consistency) conditions pick up $\delta(R)(L%
\epsilon(R))^2=o(1)$ automatically. The bridge condition that regime \textup{%
(i)} requires for parametric-rate inference is therefore 
\begin{equation*}
\sqrt{NT}\;\delta(R)\,\epsilon(R)^2\;\longrightarrow\;0
\end{equation*}
along the chosen sequence of truncation orders, which shrinks with mass and
squared diameter jointly. Under this condition \citet{freeman2026inference}
establish $\sqrt{NT}$-asymptotic normality of the second-step estimator in
the semiparametric regression model. In the construction of Remark~\ref%
{rem:blend-sharp} with a bounded density of the fixed effects, $\delta(R)$
and $\epsilon(R)$ are both of order $2^{-R} $, so the display holds once $R$
grows like $\log NT$.

No universal rate exists, however: the Dini argument is rate-free, and the
same construction can be slowed so that $\delta(R)$ and $\epsilon(R)$ vanish
arbitrarily slowly. Under regime \textup{(i)} the display therefore remains
a genuine restriction tying the sequence of truncation orders to the unknown
conflation geometry near $Z$; and near $Z$ the inherited derivatives of $h$
degenerate, so smoothness conditions on $h$ are likewise $\Omega_\delta$%
-relative. Regime \textup{(ii)} removes all of this at once, since $%
\delta(R)=\epsilon(R)=0$ for $R\ge R^{*}$.

\subsection{Additive unobserved heterogeneity}

\label{sec:additive} 

By Remark~\ref{rem:kernel}, estimating the normalised regression $h$ of
Theorem~\ref{th:proxyID} on the proxy coordinates is a $(d_{\alpha}+d_{%
\gamma})$-dimensional problem for every $R\geq R^{\ast }$ under Assumption~%
\ref{ass:multi-index} and its $\gamma$-analogue, but $d_{\alpha}+d_{\gamma}$
itself may be large. This subsection shows that when the unobserved
heterogeneity is \emph{additive} across the coordinates of the fixed
effects, the entire construction respects that structure. The singular
functions themselves are additive, their coordinate components are injective
coordinate-by-coordinate, and the normalised regression decomposes into
components each depending on one coordinate of the unit fixed effect and one
coordinate of the time fixed effect. This is the identification content of
Theorems~4 and~5 of \citet{freeman2026inference}, which uses backfitting for 
$h$.

\begin{assumption}[Additive unobserved heterogeneity on a product support]
\label{ass:additive} $d_{\alpha}=d_{\gamma}=d$; the supports are products of
compact intervals, $\Omega_{\alpha}=\prod_{k=1}^{d}\Omega_{\alpha,k}$ and $%
\Omega_{\gamma}=\prod_{k=1}^{d}\Omega_{\gamma,k}$; and 
\begin{equation*}
g(\alpha,\gamma)=\sum_{k=1}^{d}g_k(\alpha_k,\gamma_k) \qquad\text{for
functions } g_k:\Omega_{\alpha,k}\times\Omega_{\gamma,k}\to\mathbb{R}^{d_G} .
\end{equation*}
\end{assumption}

\noindent A product of compact intervals is compact, convex, and connected,
so Assumption~\ref{ass:additive} is compatible with Assumption~\ref%
{ass:domain}. No smoothness of the $g_k$ is assumed: each $g_k$ equals a
bivariate slice of $g$ net of a constant (fix a reference point $%
(\alpha^{\circ},\gamma^{\circ})\in\Omega_{\alpha}\times\Omega_{\gamma}$ and
vary only the $k$-th coordinates), so the $g_k$ inherit the continuity and
differentiability of $g$ and are identified up to additive constants. Write $%
\pi_{\gamma,k}$ for the $k$-th marginal of $\pi_\gamma$ and $\bar
v_{r,k}(c):=\mathbb{E}[v_r(\gamma_t)\mid\gamma_{t,k}=c]$; no independence
across the coordinates of $\alpha_i$ or of $\gamma_t$ is imposed.

\begin{theorem}[Additive extension]
\label{thm:additive} Suppose Assumptions~\ref{ass:domain}, \ref{ass:obseq}, %
\ref{ass:multi-index}, their $\gamma$-analogues, and Assumption~\ref%
{ass:additive} hold, let the set $\{g(\alpha_i,\gamma_t)\}$ be identified,
let Assumption~\ref{ass:S} hold with $\sigma_1>\cdots>\sigma_R>\sigma_{R+1}$%
, and let $R\ge R^{*}$ with $R^{*}$ as in Theorem~\ref{thm:headline}\textup{%
(ii)}. With $u_{r,k}$ and $v_{r,k}$ the additive components of $u_r$ and $%
v_r $ given by Proposition~\ref{prop:additive-eigen} in Appendix~\ref%
{app:inj-additive}, define the coordinate truncations 
\begin{align*}
U^{(k)}_R(s)&:=\bigl(u_{1,k}(s)^{\top},\dots,u_{R,k}(s)^{\top}\bigr)^{\top},
\qquad s\in\Omega_{\alpha,k}, \\
V^{(k)}_R(c)&:=\bigl(v_{1,k}(c),\dots,v_{R,k}(c)\bigr)^{\top}, \qquad
c\in\Omega_{\gamma,k},
\end{align*}
taking values in $\mathbb{R}^{d_G R}$ and $\mathbb{R}^{R}$, respectively.
Then:

\begin{enumerate}
\item[\textup{(i)}] each $U^{(k)}_R$ is injective on $\Omega_{\alpha,k}$ and
each $V^{(k)}_R$ is injective on $\Omega_{\gamma,k}$;

\item[\textup{(ii)}] the stacked proxies $\lambda_i:=\bigl(%
U^{(1)}_R(\alpha_{i,1}),\dots, U^{(d)}_R(\alpha_{i,d})\bigr)$ and $f_t:=%
\bigl(V^{(1)}_R(\gamma_{t,1}),\dots, V^{(d)}_R(\gamma_{t,d})\bigr)$ are
identified, continuous, and injective in $\alpha_i$ and $\gamma_t$, so they
are valid proxies in the sense of Definition~\ref{def:proxyID};

\item[\textup{(iii)}] the normalised regression $h$ of Theorem~\ref%
{th:proxyID}, formed with the stacked maps $U:=(U^{(1)}_R,\dots,U^{(d)}_R)$
and $V:=(V^{(1)}_R,\dots,V^{(d)}_R)$ of (ii), is additive: with $%
\lambda_{i,k}:=U^{(k)}_R(\alpha_{i,k})$ and $f_{t,k}:=V^{(k)}_R(%
\gamma_{t,k}) $, 
\begin{equation*}
h(\lambda_i,f_t)=g\bigl(U^{-1}(\lambda_i),V^{-1}(f_t)\bigr) %
=\sum_{k=1}^{d}h_k(\lambda_{i,k},f_{t,k}), \qquad h_k:=g_k\circ\bigl(%
(U^{(k)}_R)^{-1},(V^{(k)}_R)^{-1}\bigr),
\end{equation*}
with each $h_k$ continuous on the compact product $U^{(k)}_R(\Omega_{%
\alpha,k})\times V^{(k)}_R(\Omega_{\gamma,k})$ of two curves. Each component
depends on one coordinate of the unit fixed effect and one coordinate of the
time fixed effect, whatever $d$ and $R $.
\end{enumerate}
\end{theorem}

The proofs, and the additive structure of the singular functions on which
they rest (Proposition~\ref{prop:additive-eigen}), are in Appendix~\ref%
{app:inj-additive}.

\section{Conclusion}


Theorem~\ref{thm:headline} replaces a global univalence problem by a local
one. Under Assumption~\ref{ass:obseq} the global injectivity of the full
singular function map is \emph{free} (Proposition~\ref{prop:part1}), so the
only substantive assumption is \emph{local} immersion (Assumption~\ref%
{ass:multi-index}). The refinement to a finite truncation is then purely
topological, via monotonicity of the separation function $\Delta_R$ in %
\eqref{eq:sep} and Dini's theorem on the compact $\Omega_{\alpha}\times%
\Omega_{\alpha}$. How the truncation order $R^{*}$ depends on the floor $c_0$
in \eqref{eq:c0} and on the decay of the singular values is left open.

\appendix

\numberwithin{theorem}{section} \numberwithin{lemma}{section} %
\numberwithin{proposition}{section} \numberwithin{corollary}{section} %
\numberwithin{remark}{section} \numberwithin{assumption}{section} %
\numberwithin{equation}{section}

\section{Proofs for Sections~\protect\ref{sect:firststep} and~\protect\ref%
{sect:secondstep}}

\label{app:proofs-front}

The squared approximation error of $\mathcal{G}_{NT}$ in Section~\ref%
{sect:firststep} is the tail $\sum_{k>R_{N,T}}s_{k}^{2}/NT$, and it vanishes
as $R_{N,T}\rightarrow \infty $. For any fixed $M$ and all $N,T$ with $%
R_{N,T}\geq M$ the tail is at most $\sum_{k>M}s_{k}^{2}/NT$, which by the
theorem of \citet{eckart1936approximation} is at most $\left\Vert \Gamma
_{0,N,T}-\{g_{M}(\alpha _{0,i},\gamma _{0,t})\}\right\Vert _{NT}^{2}$, the
empirical mean of $(g_{0}-g_{M})^{2} $ over the realised fixed effects, and
this converges under Assumption~\ref{ass:S} to $\left\Vert
g_{0}-g_{M}\right\Vert _{L_{\pi }^{2}}^{2}=\sum_{r>M}\sigma _{r}^{2}$, since 
$g_{0}-g_{M}$ is bounded and continuous. Letting $M$ grow gives the claim.

\begin{proof}[Proof of Lemma~\protect\ref{lem:pseudotrue}]
Matrices of rank at most $R_{N,T}$ form a closed subset of $\mathbb{R}%
^{N\times T}$, the common zero set of all minors of order $R_{N,T}+1$, and
intersecting it with the closed ball $\left\Vert \Gamma \right\Vert
_{NT}\leq B$ gives a closed and bounded subset of a finite-dimensional
space, hence is compact. The space $\Theta _{N,T}$ is compact, so the
product is compact, and a continuous function on a compact set attains its
minimum.
\end{proof}

\begin{proof}[Proof of Theorem~\protect\ref{thm:firststep}]
Let $(\theta _{N,T}^{\ast },\Gamma _{N,T}^{\ast })$ be the minimiser of
Lemma~\ref{lem:pseudotrue}. Fix $\epsilon >0$, and let $\eta >0$ be the
bound that (ii) supplies for it, which does not depend on the realisation.
Let $E_{N,T}$ be the event on which $Q_{N,T}(\theta ,\Gamma )-Q_{N,T}(\theta
_{N,T}^{\ast },\Gamma _{N,T}^{\ast })\geq \eta $ whenever $\left\Vert
(\theta ,\Gamma )-(\theta _{N,T}^{\ast },\Gamma _{N,T}^{\ast })\right\Vert
\geq \epsilon $ on the sieve, $\Vert \vartheta _{N,T}^{\ast }-\vartheta
_{0,N,T}\Vert <\epsilon $, $\sup |\hat{Q}_{N,T}-Q_{N,T}|<\eta /4$, and the
approximate-minimiser slack is below $\eta /4$. Since $\eta $ is fixed, $%
E_{N,T}$ has probability tending to one by (i), (ii) and (iii). On $E_{N,T}$,%
\begin{align*}
Q_{N,T}(\hat{\theta}_{N,T},\hat{\Gamma}_{N,T})& \leq \hat{Q}_{N,T}(\hat{%
\theta}_{N,T},\hat{\Gamma}_{N,T})+\tfrac{\eta }{4}\leq \hat{Q}_{N,T}(\theta
_{N,T}^{\ast },\Gamma _{N,T}^{\ast })+\tfrac{\eta }{2} \\
& \leq Q_{N,T}(\theta _{N,T}^{\ast },\Gamma _{N,T}^{\ast })+\tfrac{3\eta }{4}%
<Q_{N,T}(\theta _{N,T}^{\ast },\Gamma _{N,T}^{\ast })+\eta ,
\end{align*}%
so $\left\Vert (\hat{\theta}_{N,T},\hat{\Gamma}_{N,T})-(\theta _{N,T}^{\ast
},\Gamma _{N,T}^{\ast })\right\Vert <\epsilon $, and by the triangle
inequality $\Vert (\hat{\theta}_{N,T},\hat{\Gamma}_{N,T})-\vartheta
_{0,N,T}\Vert <2\epsilon $ on $E_{N,T}$. As $\epsilon $ is arbitrary, this
proves the claim. This is the argument of Theorem 3.1 of %
\citet{chen2007large} and Theorem 2.1 of \citet{newey1994large}, with the
objective and the parameter space allowed to change with the sample.
\end{proof}

\begin{proof}[Proof of Theorem~\protect\ref{th:proxyID}]
Since $U$ and $V$ are injective, $U^{-1}(\lambda _{i})=\alpha _{i}$ and $%
V^{-1}(f_{t})=\gamma _{t}$, so $h(\lambda _{i},f_{t})=g(\alpha _{i},\gamma
_{t})$ at every such realised pair, and $h$ is determined at $\{(\lambda
_{i},f_{t})\}$ by the identified set and the identified proxies. A Borel
injection between Borel subsets of Euclidean spaces has Borel image and
Borel inverse, the Lusin--Suslin theorem, so $U^{-1}$ and $V^{-1}$ are Borel
and $h$ is Borel on $U(\Omega _{\alpha })\times V(\Omega _{\gamma })$. For
the extension, every relatively open subset of $\Omega _{\alpha }$ has
positive $\pi _{\alpha }$ mass, since $\Omega _{\alpha }$ is the support of $%
\pi _{\alpha }$, so under Assumption~\ref{ass:S} it eventually contains
realised fixed effects, and the realised fixed effects are dense in $\Omega
_{\alpha }$ in the limit, and since a continuous injection of a compact set
is a homeomorphism onto its image, the realised proxies are dense in the
compact image $U(\Omega _{\alpha })\times V(\Omega _{\gamma })$. The
function $h$ is continuous there, as a composition of continuous maps, and
is therefore determined by its values on the dense set of realised pairs,
which are identified. For the last claim, an injective $C^{q}$ immersion of
a compact set is a homeomorphism onto its image, a $C^{q}$ submanifold, with 
$C^{q}$ inverse there by the rank theorem, so $h=g\circ (U^{-1},V^{-1})$ is $%
C^{q}$.
\end{proof}

\section{Proof of Lemma~\protect\ref{lem:proxies-identified}}

\label{app:proofs-setup}

\begin{proof}[Proof of Lemma~\protect\ref{lem:proxies-identified}]
Take $d_{G}=1$, the vector case being identical with $\mathcal{H}_{\alpha }$
in place of $L_{\pi }^{2}(\Omega _{\alpha })$. Write $\pi _{\alpha ,N}$ and $%
\pi _{\gamma ,T}$ for the empirical measures of the fixed effects, $%
\left\Vert \phi \right\Vert _{N}^{2}:=\int \phi ^{2}\,d\pi _{\alpha ,N}$,
and 
\begin{equation*}
k(\alpha ,\alpha ^{\prime }):=\int_{\Omega _{\gamma }}g(\alpha ,\gamma
)g(\alpha ^{\prime },\gamma )\,d\pi _{\gamma }(\gamma ),\qquad k_{T}(\alpha
,\alpha ^{\prime }):=\int_{\Omega _{\gamma }}g(\alpha ,\gamma )g(\alpha
^{\prime },\gamma )\,d\pi _{\gamma ,T}(\gamma ),
\end{equation*}%
so that $k=\sum_{r}\sigma _{r}^{2}u_{r}\otimes u_{r}$ is the kernel of $%
T_{g}T_{g}^{\ast }$ on $L_{\pi }^{2}(\Omega _{\alpha })$, with eigenpairs $%
(\sigma _{r}^{2},u_{r})$. The matrix $[k_{T}(\alpha _{i},\alpha
_{j})]_{i,j\leq N}$ equals $T^{-1}\Gamma _{0,N,T}\Gamma _{0,N,T}^{\top }$.

Since $g$ is continuous on the compact product (Assumption~\ref{ass:domain}%
), the family $\{g(\alpha ,\cdot )g(\alpha ^{\prime },\cdot ):(\alpha
,\alpha ^{\prime })\in \Omega _{\alpha }^{2}\}$ is uniformly bounded and
equicontinuous on $\Omega _{\gamma }$, hence relatively compact in $C(\Omega
_{\gamma })$ by Arzel\`{a} and Ascoli. Weak convergence $\pi _{\gamma
,T}\Rightarrow \pi _{\gamma }$ is uniform over relatively compact subsets of 
$C(\Omega _{\gamma })$, so $\sup_{\alpha ,\alpha ^{\prime }}|k_{T}(\alpha
,\alpha ^{\prime })-k(\alpha ,\alpha ^{\prime })|\rightarrow 0$.

On $C(\Omega _{\alpha })$ define 
\begin{equation*}
(\mathcal{K}\phi )(\alpha ):=\int k(\alpha ,\alpha ^{\prime })\phi (\alpha
^{\prime })\,d\pi _{\alpha }(\alpha ^{\prime }),\qquad (\mathcal{K}%
_{N,T}\phi )(\alpha ):=\int k_{T}(\alpha ,\alpha ^{\prime })\phi (\alpha
^{\prime })\,d\pi _{\alpha ,N}(\alpha ^{\prime }).
\end{equation*}%
$\mathcal{K}_{N,T}$ is the Nystr\"{o}m approximation of $\mathcal{K}$ with
the empirical measure as quadrature rule. Its nonzero eigenvalues are those
of the $N\times N$ matrix $(NT)^{-1}\Gamma _{0,N,T}\Gamma _{0,N,T}^{\top }$,
namely $\tilde{s}_{k}^{2}/(NT)$, and if $\phi $ is an eigenfunction of $%
\mathcal{K}_{N,T}$ with eigenvalue $\mu >0$ then $(\phi (\alpha _{1}),\dots
,\phi (\alpha _{N}))$ is the corresponding eigenvector, while conversely an
eigenvector $\tilde{u}$ extends to the eigenfunction $\phi (\alpha )=\mu
^{-1}N^{-1}\sum_{j}k_{T}(\alpha ,\alpha _{j})\sqrt{N}\tilde{u}_{j}$, whose
values at the nodes are $\sqrt{N}\tilde{u}_{i}$ and which has $\left\Vert
\phi \right\Vert _{N}=\left\Vert \tilde{u}\right\Vert =1$. Two facts hold.
First, $\mathcal{K}_{N,T}\phi \rightarrow \mathcal{K}\phi $ in $C(\Omega
_{\alpha })$ for every $\phi \in C(\Omega _{\alpha })$. Replacing $k_{T}$ by 
$k$ costs at most $\sup |k_{T}-k|\,\left\Vert \phi \right\Vert _{\infty }$,
which vanishes by the uniform convergence of $k_{T}$ to $k$ above, and $\int
k(\cdot ,\alpha ^{\prime })\phi (\alpha ^{\prime })\,d\pi _{\alpha
,N}\rightarrow \int k(\cdot ,\alpha ^{\prime })\phi (\alpha ^{\prime
})\,d\pi _{\alpha }$ uniformly on $\Omega _{\alpha }$, because $\{k(\alpha
,\cdot )\phi :\alpha \in \Omega _{\alpha }\}$ is again uniformly bounded and
equicontinuous. Second, the family $\{\mathcal{K}_{N,T}\}$ is collectively
compact, that is $\bigcup_{N,T}\mathcal{K}_{N,T}(B) $ is relatively compact
in $C(\Omega _{\alpha })$ for the unit ball $B$, since these images are
bounded by $\sup |k_{T}|$ and share the modulus of continuity of $k_{T}$ in
its first argument, both uniform in $(N,T)$ by that uniform convergence.

For a pointwise convergent and collectively compact sequence of operators
the spectrum converges to that of the compact limit with multiplicity, and
for each simple eigenvalue $\sigma _{r}^{2}$ of $\mathcal{K}$ with
eigenfunction $u_{r}$ there are eigenpairs $(\mu _{r,N,T},\phi _{r,N,T})$ of 
$\mathcal{K}_{N,T}$ with $\mu _{r,N,T}\rightarrow \sigma _{r}^{2}$ and $%
\Vert \phi _{r,N,T}-u_{r}\Vert _{\infty }\rightarrow 0$ up to a sign, once $%
\phi _{r,N,T}$ is normalised to $\left\Vert \phi _{r,N,T}\right\Vert _{N}=1$%
, see \citet{anselone1971collectively} and 
\citet[Chapter
5]{chatelin2011spectral}, and \citet{koltchinskii2000random} for the sampled
case with rates. The strict ordering makes each of the leading $R$
eigenvalues simple with a gap below it, so for $N,T$ large the $r $-th
eigenpair of $\mathcal{K}_{N,T}$ is the one converging to $(\sigma
_{r}^{2},u_{r})$, and its restriction to the nodes is $\sqrt{N}\tilde{u}_{r}$%
. The normalisations match, since $\left\Vert u_{r}\right\Vert
_{N}\rightarrow \left\Vert u_{r}\right\Vert _{L_{\pi }^{2}}=1$ by weak
convergence applied to the continuous $u_{r}^{2}$. The eigenvector $\tilde{u}%
_{r}$, and with it $\phi _{r,N,T}$, is determined only up to sign. Let $%
\varsigma _{r}$ be the sign for which $\varsigma _{r}\phi _{r,N,T}$
converges to $u_{r}$ rather than to $-u_{r}$, and write $\tilde{u}_{r}$ and $%
\phi _{r,N,T}$ for $\varsigma _{r}\tilde{u}_{r}$ and $\varsigma _{r}\phi
_{r,N,T}$ from here on. Hence $\tilde{s}_{r}/\sqrt{NT}=\mu
_{r,N,T}^{1/2}\rightarrow \sigma _{r}$ and $\max_{i\leq N}|\sqrt{N}\tilde{u}%
_{r,i}-u_{r}(\alpha _{i})|\leq \left\Vert \phi _{r,N,T}-u_{r}\right\Vert
_{\infty }\rightarrow 0$.

In the singular value decomposition $\tilde{s}_{r}\tilde{v}_{r}=\Gamma
_{0,N,T}^{\top }\tilde{u}_{r}$, which fixes the sign of $\tilde{v}_{r}$
given that of $\tilde{u}_{r}$ exactly as $\sigma _{r}v_{r}=T_{g}^{\ast
}u_{r} $ fixes the sign of $v_{r}$ given that of $u_{r}$. Therefore 
\begin{equation*}
\sqrt{T}\,\tilde{v}_{r,t}=\frac{\sqrt{NT}}{\tilde{s}_{r}}\int g(\alpha
,\gamma _{t})\sqrt{N}\tilde{u}_{r}(\alpha )\,d\pi _{\alpha ,N}(\alpha
)\rightarrow \frac{1}{\sigma _{r}}\int g(\alpha ,\gamma _{t})u_{r}(\alpha
)\,d\pi _{\alpha }(\alpha )=v_{r}(\gamma _{t}),
\end{equation*}%
where $\sqrt{N}\tilde{u}_{r}(\alpha _{i}):=\sqrt{N}\tilde{u}_{r,i}$, using
the eigenpair convergence above for the first factor and for the uniform
replacement of $\sqrt{N}\tilde{u}_{r,i}$ by $u_{r}(\alpha _{i})$, and weak
convergence of $\pi _{\alpha ,N}$ over the uniformly bounded equicontinuous
family $\{g(\cdot ,\gamma )u_{r}(\cdot ):\gamma \in \Omega _{\gamma }\}$ for
the limit, uniformly in $t$.

$\Gamma _{0,N,T}$ is identified, so some $\tilde{\Gamma}_{N,T}$ determined
by the distribution of the data has $\Vert \tilde{\Gamma}_{N,T}-\Gamma
_{0,N,T}\Vert _{NT}\rightarrow 0$. By the inequality of Weyl the singular
values of $\tilde{\Gamma}_{N,T}/\sqrt{NT}$ lie within $\Vert \tilde{\Gamma}%
_{N,T}-\Gamma _{0,N,T}\Vert _{NT}$ of the $\tilde{s}_{k}/\sqrt{NT}$, and
since the limits $\sigma _{1}>\cdots >\sigma _{R+1}$ are distinct, the
perturbation bound of \citet{wedin1972perturbation}, applied at each of
these gaps, puts the leading $R$ singular vectors of $\tilde{\Gamma}_{N,T}$,
with suitable signs, within $o(1)$ of $\tilde{u}_{r}$ and $\tilde{v}_{r}$ in
the Euclidean norm. Combined with the convergence established above, $\sigma
_{r}$, $u_{r}(\alpha _{i})$ and $v_{r}(\gamma _{t})$ for $r\leq R$ are
identified, the last two up to a common sign and in the empirical $L^{2}$
norms over units and over periods. The values $u_{r}(\alpha _{i})$ and $%
v_{r}(\gamma _{t})$ are the coordinates of $\lambda _{i}=U_{R}(\alpha _{i})$
and $f_{t}=V_{R}(\gamma _{t})$. If some of the leading $R$ singular values
coincide, the same argument identifies the eigenspaces of the tied blocks,
so $U_{R}$ and $V_{R}$ up to an orthogonal rotation within each block, which
suffices for their injectivity.
\end{proof}

\section{Technical results and proofs for Section~\protect\ref%
{sect:mainResult}}

\label{app:injectivity}

This appendix collects the lemmas behind the injectivity results of Section~%
\ref{sect:mainResult} and their proofs.

\subsection{Response map and gradient information}

\label{app:inj-basic}

\begin{lemma}
\label{lem:Phi-cont} Under Assumption~\ref{ass:domain}, $\Phi$ is Lipschitz
from $\Omega_{\alpha}$ into $\mathcal{H}_{\gamma}$, $\lVert
\Phi(\alpha_1)-\Phi(\alpha_2)\rVert_{\mathcal{H}_{\gamma}}\le L\,\lVert
\alpha_1-\alpha_2\rVert$ with $L=\sup_{\Omega_{\alpha}\times\Omega_{\gamma}}%
\lVert J\rVert<\infty$. Hence $\Delta(\alpha_1,\alpha_2):=\lVert
\Phi(\alpha_1)-\Phi(\alpha_2)\rVert_{\mathcal{H}_{\gamma}}^2$ is continuous
on $\Omega_{\alpha}\times\Omega_{\alpha}$.
\end{lemma}

\begin{proof}
The increment admits the integral form $g(\alpha_1,\gamma)-g(\alpha_2,%
\gamma) =\int_0^1 J(\alpha_2+t(\alpha_1-\alpha_2),\gamma)
(\alpha_1-\alpha_2)\,dt$; since $J$ is bounded on the compact $%
\Omega_{\alpha}\times\Omega_{\gamma}$, taking $\mathcal{H}_{\gamma}$ norms
gives the bound.
\end{proof}

\begin{lemma}
\label{lem:gradexp} Under Assumption~\ref{ass:domain}, for every $%
\alpha\in\Omega_{\alpha}$:

\begin{enumerate}
\item[\textup{(i)}] the differentiated series reproduces the Jacobian, 
\begin{equation}
J(\alpha,\cdot)=\sum_{r = 1}^\infty\sigma_r\,Du_r(\alpha)\,v_r, \qquad\text{%
with convergence entrywise in }L^{2}_{\pi}(\Omega_{\gamma});
\label{eq:gradseries}
\end{equation}

\item[\textup{(ii)}] consequently $\sum_{r\ge1}\sigma_r^2\lVert
Du_r(\alpha)\rVert^2 \le\int_{\Omega_{\gamma}}\lVert
J(\alpha,\gamma)\rVert^2d\pi_\gamma<\infty$, the monotone limit $%
M_\infty(\alpha)=\sum_{r\ge1}\sigma_r^2 Du_r(\alpha)^{\top} Du_r(\alpha)$ is
finite, and it equals the closed form \eqref{eq:Minf}.
\end{enumerate}
\end{lemma}

\begin{proof}
It suffices to establish \eqref{eq:fourier}, and with it (i)--(ii), for $%
\alpha$ in the interior of $\Omega_{\alpha}$, which is dense in $%
\Omega_{\alpha}$ (a compact convex set of positive Lebesgue measure,
guaranteed by the density bounds, equals the closure of its interior); the
extension to boundary points follows by continuity, since $Du_r$ and $%
\alpha\mapsto D_\alpha g(\alpha,\cdot)$ are continuous and the subspace $%
\mathcal{V}$ below is closed. Fix, then, $\alpha$ in the interior of $%
\Omega_{\alpha}$. Because $D_\alpha g$ is jointly continuous on the compact $%
\Omega_{\alpha}\times\Omega_{\gamma}$ it is bounded there, so
differentiation under the integral sign is justified in $u_r(\alpha)=%
\sigma_r^{-1}\int_{\Omega_{\gamma}}g(\alpha,\gamma)v_r(\gamma)\,d\pi_\gamma$%
, giving, entrywise, 
\begin{equation}
\int_{\Omega_{\gamma}}J(\alpha,\gamma)\,v_r(\gamma)\,d\pi_\gamma
=\sigma_r\,Du_r(\alpha)\ \in\mathbb{R}^{d_G\times\,d_{\alpha}}.
\label{eq:fourier}
\end{equation}
Fix an entry $(i,j)$. By \eqref{eq:fourier}, $\{\sigma_r\,\partial_{%
\alpha_j}u_{r,i}(\alpha)\}_{r\ge1}$ are the Fourier coefficients of $%
\partial_{\alpha_j}g_i(\alpha,\cdot)\in\,L^{2}_{\pi}(\Omega_{\gamma})$
against the orthonormal system $\{v_r\}$, so the scalar Bessel inequality
gives 
\begin{equation*}
\sum_{r\ge1}\sigma_r^2\,\bigl(\partial_{\alpha_j}u_{r,i}(\alpha)\bigr)^2
\le\lVert
\partial_{\alpha_j}g_i(\alpha,\cdot)\rVert_{L^{2}_{\pi}(\Omega_{\gamma})}^2
=\int_{\Omega_{\gamma}}\bigl(\partial_{\alpha_j}g_i(\alpha,\gamma)\bigr)%
^2\,d\pi_\gamma .
\end{equation*}
Now sum over the $d_Gd_{\alpha}$ entries. Using the Frobenius identities $%
\lVert Du_r(\alpha)\rVert^2=\sum_{i,j}\bigl(\partial_{\alpha_j}u_{r,i}(%
\alpha) \bigr)^2$ and $\lVert J(\alpha,\gamma)\rVert^2 =\sum_{i,j}\bigl(%
\partial_{\alpha_j}g_i(\alpha,\gamma)\bigr)^2$, this gives 
\begin{equation*}
\sum_{r\ge1}\sigma_r^2\lVert Du_r(\alpha)\rVert^2
\le\int_{\Omega_{\gamma}}\lVert J(\alpha,\gamma)\rVert^2\,d\pi_\gamma
\le\sup_{\Omega_{\alpha}\times\Omega_{\gamma}}\lVert J\rVert^2<\infty,
\end{equation*}
the final bound because $\pi_\gamma$ is a probability measure. This is the
summability in (ii); in particular the monotone psd limit $M_\infty(\alpha)$
is finite.

It remains to identify the limit of $J(\alpha,\cdot)$. Let $\mathcal{V}:=%
\overline{\spn}\{v_r:r\ge1\}\subseteq\,L^{2}_{\pi}(\Omega_{\gamma})$, an
orthonormal basis of which is $\{v_r\}$, and let $\mathcal{V}%
^{\,d_G}:=\{\phi\in\mathcal{H}_{\gamma}:\phi_i\in\mathcal{V}\text{ for each }%
i\}$, the $d_G$-fold product of $\mathcal{V}$ and a closed subspace of $%
\mathcal{H}_{\gamma}$. Since the SVD \eqref{eq:svd} converges in $%
L^{2}_{\pi}(\Omega_{\alpha}\times\Omega_{\gamma};\mathbb{R}^{d_G})$, we have 
$g(\alpha,\cdot)\in\mathcal{V}^{\,d_G}$ for $\pi_\alpha$-almost every $%
\alpha $; as $\alpha\mapsto g(\alpha,\cdot)$ is continuous into $\mathcal{H}%
_{\gamma}$ (Lemma~\ref{lem:Phi-cont}) and $\mathcal{V}^{\,d_G}$ is closed,
in fact $g(\alpha,\cdot)\in\mathcal{V}^{\,d_G}$ for \emph{every} $\alpha$.
For a coordinate direction $e_j$ the difference quotients $h^{-1}\bigl(%
g(\alpha+he_j,\cdot)-g(\alpha,\cdot)\bigr)$ lie in the subspace $\mathcal{V}%
^{\,d_G}$ and, by the mean value theorem together with dominated convergence
(this is dominant since the constant $\sup_{\Omega_{\alpha}\times\Omega_{%
\gamma}}\lVert J\rVert$ is integrable since $\pi_\gamma$ is finite),
converge in $\mathcal{H}_{\gamma}$ to $\partial_{\alpha_j} g(\alpha,\cdot)$.
As $\mathcal{V}^{\,d_G}$ is closed, $\partial_{\alpha_j}g(\alpha,\cdot)\in%
\mathcal{V}^{\,d_G}$ and equals its orthogonal projection onto $\mathcal{V}%
^{\,d_G}$, namely $\sum_r\bigl(\int_{\Omega_{\gamma}}\partial_{\alpha_j}g(%
\alpha,\cdot)v_r\,d\pi_\gamma \bigr)v_r$, which by \eqref{eq:fourier} is the 
$j$-th column of $\sum_r\sigma_r Du_r(\alpha)v_r$. Collecting the columns
proves (i). Finally, expanding $M_\infty(\alpha)=\sum_r\sigma_r^2
Du_r(\alpha)^{\top}Du_r(\alpha)$ and using \eqref{eq:gradseries} with
orthonormality of $\{v_r\}$ in $L^{2}_{\pi}(\Omega_{\gamma})$ yields $%
M_\infty(\alpha)=\int_{\Omega_{\gamma}}J^{\top}J\,d\pi_\gamma$, the closed
form \eqref{eq:Minf}.
\end{proof}

\begin{remark}
\label{rem:gradseries} Lemma~\ref{lem:gradexp} is the only convergence fact
used below. It follows from Assumption~\ref{ass:domain} alone and is weaker
than uniform $C^1$ convergence of \eqref{eq:svd}.
\end{remark}

\begin{lemma}
\label{lem:Minf-uniform} Under Assumption~\ref{ass:domain}, $\alpha\mapsto
M_\infty(\alpha)$ is continuous on $\Omega_{\alpha}$. If in addition
Assumption~\ref{ass:multi-index} holds, then $M_\infty(\alpha)\succ0$ for
every $\alpha\in\Omega_{\alpha}$ and, where $\lambda_{\min}$ denotes minimum
eigenvalue, 
\begin{equation}
c_0:=\min_{\alpha\in\Omega_{\alpha}}\lambda_{\min}\bigl(M_\infty(\alpha)%
\bigr)>0, \qquad M_\infty(\alpha)\succeq c_0 I_{d_{\alpha}}\ \text{ for all }%
\alpha\in\Omega_{\alpha}.  \label{eq:c0}
\end{equation}
\end{lemma}

\begin{proof}
Continuity follows from \eqref{eq:Minf} (Lemma~\ref{lem:gradexp}) and joint
continuity of $J$. For a unit vector $v$, $v^{\top}M_\infty(\alpha)v=\int_{%
\Omega_{\gamma}}\lVert J(\alpha,\gamma)v\rVert^2\, d\pi_\gamma$; were this
zero, $J(\alpha,\cdot)v$ would vanish $\pi_\gamma$-a.e., hence, by
continuity, at the evaluation points $\gamma_{0,1},\dots,\gamma_{0,p_\alpha}%
\in\Omega_{\gamma}=\text{supp}\pi_\gamma$, contradicting $\rank D_\alpha
G_1(\alpha,b_0)=d_{\alpha}$. Thus $M_\infty(\alpha)\succ0$, and $%
\lambda_{\min}(M_\infty(\cdot))$ is continuous and strictly positive on the
compact $\Omega_{\alpha}$, so attains a positive minimum $c_0$.
\end{proof}

\noindent In particular, Assumption~\ref{ass:multi-index} forces $%
Z=\varnothing$: the almost-everywhere immersion of Assumption~\ref%
{ass:ae-immersion} then holds everywhere.

\subsection{Injectivity of the full collection}

\label{app:inj-full}

\begin{proof}[Proof of Proposition~\protect\ref{prop:part1}]
Fix $\alpha\in\Omega_{\alpha}$. The element $\Phi(\alpha)=g(\alpha,\cdot)\in%
\mathcal{H}_{\gamma}$ has the expansion $\Phi(\alpha)=\sum_{r\ge1}\bigl(%
\sigma_r u_r(\alpha)\bigr)v_r$ with vector coefficients: as in the proof of
Lemma~\ref{lem:gradexp}, $\Phi(\alpha)\in\mathcal{V}^{\,d_G}$ for \emph{every%
} $\alpha$ (the $L^2$-convergence of \eqref{eq:svd} gives this $\pi_\alpha$%
-a.e., and $\Phi$ is continuous with $\mathcal{V}^{\,d_G}$ closed), so $%
\Phi(\alpha)=\sum_r c_r(\alpha)v_r$ with $c_r(\alpha)=\int_{\Omega_{\gamma}}%
\Phi(\alpha)v_r\,d\pi_\gamma=(T_gv_r)(\alpha) =\sigma_r u_r(\alpha)\in%
\mathbb{R}^{d_G}$ by the continuous versions. Since $\{v_re_i\}_{r\ge1,\,i%
\le d_G}$ is orthonormal in $\mathcal{H}_{\gamma}$, $e_1,\dots,e_{d_G}$
being the coordinate vectors of $\mathbb{R}^{d_G}$, Parseval's identity
applied to $\Phi(\alpha_1)-\Phi(\alpha_2) =\sum_r\sigma_r\bigl(%
u_r(\alpha_1)-u_r(\alpha_2)\bigr)v_r$ yields \eqref{eq:parseval}.

For the equivalence: the right-hand side of \eqref{eq:parseval} is zero if
and only if $\sigma_r\lVert u_r(\alpha_1)-u_r(\alpha_2)\rVert=0$ for every $%
r $, which, because $\sigma_r>0$, holds if and only if $u_r(\alpha_1)=u_r(%
\alpha_2)$ for every $r$, i.e.\ $U_\infty(\alpha_1)= U_\infty(\alpha_2)$.
The left-hand side is zero if and only if $\Phi(\alpha_1)=\Phi(\alpha_2)$.
Therefore $\Phi(\alpha_1)=\Phi(\alpha_2)$ iff $U_\infty(\alpha_1)=U_\infty(%
\alpha_2)$, and (i)$\Leftrightarrow$(ii) follows.
\end{proof}

\begin{proof}[Proof of Corollary~\protect\ref{cor:part1}]
Immediate from Proposition~\ref{prop:part1}, (i)$\Rightarrow$(ii).
\end{proof}

\begin{remark}
\label{rem:positive-sigma} The collection $\{u_r\}$ ranges over the positive
singular values, equivalently over an orthonormal basis of $\overline{\text{%
range}(T_g)}$, which is what ``$\sigma_r>0$ throughout'' records. This is a
convention, not a restriction: a finite-rank $g$ ($\sigma_r=0$ for $r>M$) is
covered with $U_\infty=U_M=(u_1,\dots,u_M)$, the proof unchanged since the
dropped terms vanish in \eqref{eq:parseval}. Positivity cannot simply be
waived, however: were zero-$\sigma$ directions (an arbitrary orthonormal
completion of $\mathcal{H}_{\alpha}$) appended to the collection, the step ``%
$\sigma_r(u_r(\alpha_1)-u_r(\alpha_2))=0$ for all $r$ $\Rightarrow
u_r(\alpha_1)=u_r(\alpha_2)$ for all $r$'' would fail and \textup{(ii)}$%
\Rightarrow$\textup{(i)} would break, as such directions carry no
information about $g$. Finally, when $M<\infty$ one has $\Delta_R=\Delta$
for all $R\ge M$, so Theorem~\ref{thm:main} holds with $R_0=M$ directly from
Assumption~\ref{ass:obseq}, with no recourse to Assumption~\ref%
{ass:multi-index}.
\end{remark}

\subsection{Finite truncations}

\label{app:inj-finite}

We now show that a finite leading truncation $U_R$ inherits injectivity once 
$R$ is large enough. This is the substantive step: injectivity of the
infinite collection (Corollary~\ref{cor:part1}) does not by itself guarantee
that any finite truncation separates all pairs of points, and an additional
local condition is required.

It is convenient to introduce, for $R\in\mathbb{N}\cup\{\infty\}$, the \emph{%
separation function} 
\begin{equation}
\Delta_R(\alpha_1,\alpha_2) :=\sum_{r\le R}\sigma_r^2\,\lVert
u_r(\alpha_1)-u_r(\alpha_2)\rVert^2 =\lVert
\Phi_R(\alpha_1)-\Phi_R(\alpha_2)\rVert_{\mathcal{H}_{\gamma}}^2,
\label{eq:sep}
\end{equation}
where the second equality follows as in Proposition~\ref{prop:part1} applied
to $g_R$. Three properties are immediate:

\begin{itemize}
\item[(a)] \emph{Continuity}: each $\Delta_R$ ($R\in\mathbb{N}$) is
continuous on $\Omega_{\alpha}\times\Omega_{\alpha}$, being a finite sum of
continuous functions; and $\Delta_\infty=\Delta$ is continuous, equal to $%
\lVert \Phi(\alpha_1)-\Phi(\alpha_2)\rVert^2_{\mathcal{H}_{\gamma}}$, by
Lemma~\ref{lem:Phi-cont}.

\item[(b)] \emph{Monotonicity}: each summand $\ge0$, so $\Delta_{R}\le%
\Delta_{R+1}$ and $\Delta_R\nearrow\Delta$ pointwise as $R\to\infty$.

\item[(c)] \emph{Injectivity criterion}: since $\sigma_r>0$, $U_R$ is
injective if and only if $\Delta_R(\alpha_1,\alpha_2)>0$ for all $%
\alpha_1\ne\alpha_2$.
\end{itemize}

Under Assumption~\ref{ass:obseq} and Proposition~\ref{prop:part1}, we have $%
\Delta(\alpha_1,\alpha_2)>0$ for all $\alpha_1\ne\alpha_2$.

\begin{lemma}
\label{lem:near} Under Assumptions~\ref{ass:domain} and~\ref{ass:multi-index}%
, with $c_0$ as in \eqref{eq:c0}, there exist a finite $R_1$ and a radius $%
\epsilon>0$ such that 
\begin{equation*}
\Delta_{R_1}(\alpha_1,\alpha_2)\ \ge\ \frac{c_0}{8}\,\lVert
\alpha_1-\alpha_2\rVert^2 \qquad\text{whenever }\ \lVert
\alpha_1-\alpha_2\rVert<\epsilon .
\end{equation*}
\end{lemma}

\begin{proof}
\emph{Choice of $R_1$.} Each $M_R$ is continuous (finite sum of continuous
matrix-valued functions) and $M_R\nearrow M_\infty$ pointwise in the
positive semidefinite order (Lemma~\ref{lem:gradexp}), with $M_\infty$
continuous by Lemma~\ref{lem:Minf-uniform}. The scalar functions $%
\alpha\mapsto v^{\top}M_R(\alpha)v$ increase monotonically to $\alpha\mapsto
v^{\top}M_\infty(\alpha)v$ for each fixed unit vector $v$; by Dini's theorem
in \cite{rudin1976} this convergence is uniform on the compact $%
\Omega_{\alpha}$, and a finite cover of the unit sphere in $\mathbb{R}%
^{d_{\alpha}}$ makes the convergence $M_R\to M_\infty$ uniform in operator
norm on $\Omega_{\alpha}$. Hence we may fix $R_1$ with 
\begin{equation}
M_{R_1}(\alpha)\succeq\tfrac{c_0}{2}\,I_{d_{\alpha}} \qquad\text{for all }%
\alpha\in\Omega_{\alpha}.  \label{eq:MR1}
\end{equation}

\emph{Local quadratic lower bound.} Fix $\alpha_1,\alpha_2\in\Omega_{\alpha}$
and write $\delta=\alpha_1-\alpha_2$. By convexity of $\Omega_{\alpha}$
(Assumption~\ref{ass:domain}), the segment $\zeta(t)=\alpha_2+t\delta$ lies
in $\Omega_{\alpha}$ for $t\in[0,1]$. Let $\Sigma_{R_1}=\diag(\sigma_1
I_{d_G},\dots,\sigma_{R_1}I_{d_G})\in\mathbb{R}^{d_G R_1\times d_G R_1}$ and 
$P(t)=\Sigma_{R_1}\,DU_{R_1}(\zeta(t))\in\mathbb{R}^{d_G
R_1\times\,d_{\alpha}} $, where $DU_{R_1}=\bigl(Du_1^{\top},\dots,Du_{R_1}^{%
\top}\bigr)^{\top} \in\mathbb{R}^{d_G R_1\times\,d_{\alpha}}$ is the
Jacobian of $U_{R_1}$. By the fundamental theorem of calculus, 
\begin{equation*}
\Sigma_{R_1}\bigl(U_{R_1}(\alpha_1)-U_{R_1}(\alpha_2)\bigr) =\int_0^1
P(t)\,\delta\,dt,
\end{equation*}
so, by the triangle inequality, 
\begin{equation}
\Delta_{R_1}(\alpha_1,\alpha_2)^{1/2} =\Bigl\|\int_0^1 P(t)\delta\,dt\Bigr\| %
\ge\lVert P(0)\delta\rVert-\int_0^1\lVert P(t)-P(0)\rVert\,dt\;\lVert
\delta\rVert.  \label{eq:segment}
\end{equation}
For the first term, $\lVert P(0)\delta\rVert^2
=\delta^{\top}DU_{R_1}(\alpha_2)^{\top}\Sigma_{R_1}^2
DU_{R_1}(\alpha_2)\delta =\delta^{\top}M_{R_1}(\alpha_2)\delta\ge\tfrac{c_0}{%
2}\lVert \delta\rVert^2$ by \eqref{eq:MR1}, hence $\lVert P(0)\delta\rVert\ge%
\sqrt{c_0/2}\,\lVert \delta\rVert$. For the second term, $P$ is continuous
on the compact $\Omega_{\alpha}$ (as $DU_{R_1}$ is continuous), hence
uniformly continuous; choose $\epsilon>0$ so that $\lVert
\alpha-\alpha^{\prime }\rVert<\epsilon$ implies $\lVert
\Sigma_{R_1}DU_{R_1}(\alpha)-\Sigma_{R_1}DU_{R_1}(\alpha^{\prime })\rVert
\le\tfrac12\sqrt{c_0/2}$. Then whenever $\lVert \delta\rVert<\epsilon$ we
have $\lVert \zeta(t)-\alpha_2\rVert=t\lVert \delta\rVert<\epsilon$ and so $%
\lVert P(t)-P(0)\rVert\le\tfrac12\sqrt{c_0/2}$ for all $t$. Substituting
into \eqref{eq:segment}, 
\begin{equation*}
\Delta_{R_1}(\alpha_1,\alpha_2)^{1/2} \ge\sqrt{c_0/2}\,\lVert
\delta\rVert-\tfrac12\sqrt{c_0/2}\,\lVert \delta\rVert =\tfrac12\sqrt{c_0/2}%
\,\lVert \delta\rVert,
\end{equation*}
whence $\Delta_{R_1}(\alpha_1,\alpha_2)\ge\tfrac{c_0}{8}\lVert
\delta\rVert^2 $ for $\lVert \delta\rVert<\epsilon$, as claimed.
\end{proof}

\subsubsection{Pointwise almost-everywhere injectivity}

\label{app:ae} 

Assumption~\ref{ass:multi-index} requires full rank at \emph{every} $\alpha$%
. The example $g(\alpha,\gamma)=\alpha^3 \ell(\gamma)$ of Remark~\ref%
{rem:independence} shows this can fail on a $\pi_\alpha$-null set without
harming identification: there $U_\infty$ stays injective and the drop occurs
only at $\alpha=0$, with $\pi_\alpha(\{0\})=0$. Assumption~\ref%
{ass:ae-immersion} is the weaker notion that tolerates such a drop, which is
the singular-function-map analogue of injectivity almost everywhere (global
invertibility) in the sense of \cite{traver2025global}. Its degeneracy set $%
Z $ is closed, since $M_\infty$ is continuous (Lemma~\ref{lem:Minf-uniform}).

\begin{proposition}
\label{prop:ae} Under Assumptions~\ref{ass:domain}, \ref{ass:obseq}, and~\ref%
{ass:ae-immersion}, for every compact $K\subseteq\Omega_{\alpha}\setminus Z$
there is a finite $R_K$ with $U_R$ injective on $K$ for all $R\ge R_K$.
Hence $\pi_\alpha$-almost every $\alpha$ has a neighbourhood on which a
finite leading truncation is injective.
\end{proposition}

\begin{proof}
On compact $K$, matrix $M_\infty$ is continuous and pd, so $%
c_K:=\min_{\alpha\in K}\lambda_{\min}(M_\infty(\alpha))>0$. Lemma~\ref%
{lem:near} applies with the global floor \eqref{eq:c0} replaced by $%
M_\infty\succeq c_K I_{d_{\alpha}}$ on $K$: its proof uses the floor only at
the base point $\alpha_2$ and through the uniform convergence $M_R\to
M_\infty$ on the compact $\Omega_{\alpha}$, which is unaffected. This yields
finite $R_1$ and $\epsilon>0$ with $\Delta_{R_1}(\alpha_1,\alpha_2)>0$
whenever $0<\lVert \alpha_1-\alpha_2\rVert<\epsilon$ and $\alpha_2\in K$.
Away from the diagonal, on the compact set $\{(\alpha_1,\alpha_2)\in K\times
K:\lVert \alpha_1-\alpha_2\rVert\ge\epsilon\}$ the continuous function $%
\Delta$ is strictly positive (Proposition~\ref{prop:part1} and Assumption~%
\ref{ass:obseq}), hence attains a positive minimum, and by Dini's theorem $%
\Delta_R\nearrow\Delta$ uniformly there; choose $R_K\ge R_1$ with $%
\Delta_R>0 $ on this set for all $R\ge R_K$. By monotonicity (property (b))
and the injectivity criterion (property (c)), $U_R$ is injective on $K$ for
every $R\ge R_K$. As $Z$ is closed, $\Omega_{\alpha}\setminus Z$ is open
and, since $\pi_\alpha(Z)=0$, almost every $\alpha$ lies in such a $K$.
\end{proof}

The order $R_K$ grows as $K$ approaches $Z$. A measure-zero rank drop costs
the single uniform truncation order, restored under the stronger Assumption~%
\ref{ass:multi-index} in Theorem~\ref{thm:main} below, but not
identification. As in \cite{traver2025global}, the global condition is
supplied by Assumption~\ref{ass:obseq}, which is the analogue of prescribed
homeomorphic boundary data, and Assumption~\ref{ass:ae-immersion} replaces
nondegeneracy everywhere.


\subsubsection{Uniform finite leading singular functions}


Under the stronger Assumption~\ref{ass:multi-index} the degeneracy set is
empty (Lemma~\ref{lem:Minf-uniform}) and the compact exhaustion of
Proposition~\ref{prop:ae} is unnecessary: a single finite truncation order
delivers injectivity on all of $\Omega_{\alpha}$.

\begin{theorem}
\label{thm:main} Suppose Assumptions~\ref{ass:domain}, \ref{ass:obseq}, and~%
\ref{ass:multi-index} hold. Then there exists a finite $R_0$ such that for
every $R\ge R_0$ the leading truncation $U_R$ of \eqref{eq:UR} is injective
on $\Omega_{\alpha}$.
\end{theorem}

\begin{proof}
Let $R_1$ and $\epsilon>0$ be as in Lemma~\ref{lem:near}, so that 
\begin{equation}
\Delta_{R_1}(\alpha_1,\alpha_2)>0 \qquad\text{whenever }0<\lVert
\alpha_1-\alpha_2\rVert<\epsilon.  \label{eq:nearpos}
\end{equation}

\emph{Away from the diagonal.} The set $K_\epsilon:=\{(\alpha_1,\alpha_2)\in%
\Omega_{\alpha}\times\Omega_{\alpha}:\lVert \alpha_1-\alpha_2\rVert
\ge\epsilon\}$ is closed in the compact $\Omega_{\alpha}\times\Omega_{%
\alpha} $, hence compact. The function $\Delta$ is continuous on $K_\epsilon$
and, by Proposition~\ref{prop:part1} together with Assumption~\ref{ass:obseq}%
, strictly positive there; therefore it attains a minimum 
\begin{equation*}
m_\epsilon:=\min_{(\alpha_1,\alpha_2)\in
K_\epsilon}\Delta(\alpha_1,\alpha_2)>0.
\end{equation*}
The functions $\Delta_R$ increase monotonically to the continuous limit $%
\Delta$ on the compact $K_\epsilon$, with each $\Delta_R$ continuous; by
Dini's theorem the convergence is uniform on $K_\epsilon$. Choose $R_2\ge
R_1 $ such that $\sup_{K_\epsilon}(\Delta-\Delta_R)\le m_\epsilon/2$ for all 
$R\ge R_2$, so that 
\begin{equation}
\Delta_R(\alpha_1,\alpha_2)\ge\Delta(\alpha_1,\alpha_2)-\tfrac{m_\epsilon}{2}
\ge\tfrac{m_\epsilon}{2}>0 \qquad\text{on }K_\epsilon,\ \text{for all }R\ge
R_2.  \label{eq:farpos}
\end{equation}

\emph{Conclusion.} Set $R_0:=R_2$ and take any $R\ge R_0$ and any $%
\alpha_1\ne\alpha_2$. If $\lVert \alpha_1-\alpha_2\rVert<\epsilon$, then by
monotonicity (property (b)) and \eqref{eq:nearpos}, 
\begin{equation*}
\Delta_R(\alpha_1,\alpha_2)\ge\Delta_{R_1}(\alpha_1,\alpha_2)>0.
\end{equation*}
If $\lVert \alpha_1-\alpha_2\rVert\ge\epsilon$, then $(\alpha_1,\alpha_2)\in
K_\epsilon$ and \eqref{eq:farpos} gives $\Delta_R(\alpha_1,\alpha_2)>0$. In
both cases $\Delta_R(\alpha_1,\alpha_2)>0$. By the injectivity criterion
(property (c)), $U_R$ is injective.
\end{proof}

\subsubsection{Quantifying the trade-off}


The trade-off between Theorem~\ref{thm:main} and Proposition~\ref{prop:ae}
can be quantified: a single finite truncation order suffices on all but an
arbitrarily small mass of fixed effects.

\begin{corollary}
\label{cor:blend} Under Assumptions~\ref{ass:domain}, \ref{ass:obseq}, and~%
\ref{ass:ae-immersion}, for every $\delta>0$ there exist a finite $R(\delta)$
and a compact $\Omega_\delta\subseteq\Omega_{\alpha}\setminus Z$ with $%
\pi_\alpha(\Omega_\delta)\ge1-\delta$ such that $U_R$ is injective on $%
\Omega_\delta$ for every $R\ge R(\delta)$.
\end{corollary}

\begin{proof}
For $\eta>0$ let $T_\eta:=\{\alpha\in\Omega_{\alpha}:\text{dist}%
(\alpha,Z)<\eta\}$. Since $Z$ is closed, $\bigcap_{\eta>0}T_\eta=Z$, so
continuity from above of the finite measure $\pi_\alpha$ gives $%
\pi_\alpha(T_\eta)\downarrow\pi_\alpha(Z)=0$ as $\eta\downarrow0$. Fix $\eta$
with $\pi_\alpha(T_\eta)\le\delta$ and set $\Omega_\delta:=\Omega_{\alpha}%
\setminus T_\eta$, a compact subset of $\Omega_{\alpha}\setminus Z$ with $%
\pi_\alpha(\Omega_\delta)\ge1-\delta$. Proposition~\ref{prop:ae} applied to $%
K=\Omega_\delta$ yields the finite $R(\delta):=R_K$.
\end{proof}

\begin{remark}
\label{rem:blend-sharp} Corollary~\ref{cor:blend} cannot be strengthened to
a null exceptional set: there exist functions $g$ satisfying Assumptions~\ref%
{ass:domain}, \ref{ass:obseq}, and~\ref{ass:ae-immersion} for which, for 
\emph{every} finite $R$, the truncation $U_R$ fails to be injective on every
set of full $\pi_\alpha$-measure. The mechanism is that conflated pairs may
accumulate at $Z$ at ever finer scales, each resolved only by ever later
singular functions. With $d_{\alpha}=1$, $d_G=1$, and $Z=\{0\}$, take $%
u_1\propto\alpha^2$ near the origin and let each $u_{k+1}$ carry an odd
component supported on $\{|\alpha|\asymp2^{-k}\}$. Then $U_\infty$ is
injective and $M_\infty(\alpha)>0$ for $\alpha\neq0$, yet for each fixed $R$
every mirror pair $(\alpha,-\alpha)$ at scales finer than $2^{-R}$ is
conflated under $U_R$, a positive-measure family of conflated pairs meeting
every full-measure set. Hence $R(\delta)\to\infty$ as $\delta\downarrow0$ is
unavoidable in general, and Corollary~\ref{cor:blend} interpolates between
Theorem~\ref{thm:main} ($\delta=0$, under the everywhere floor of Assumption~%
\ref{ass:multi-index}) and Proposition~\ref{prop:ae}.
\end{remark}

\subsection{Right singular functions and the proof of Theorem~\protect\ref%
{thm:headline}}

\label{sec:gamma} 

Both results transfer to the right singular functions $\{v_r\}$ and the time
fixed effect $\gamma$, using the \emph{same} singular system. The $\gamma$%
-side response map is $\Psi(\gamma):=g(\cdot,\gamma)\in\mathcal{H}_{\alpha}$%
, and pairing it with the orthonormal system $\{u_r\}\subset\mathcal{H}%
_{\alpha}$ gives, by the adjoint, 
\begin{equation}
\langle \Psi(\gamma),\,u_r\rangle_{\mathcal{H}_{\alpha}}
=\int_{\Omega_{\alpha}}g(\alpha,\gamma)^{\top}u_r(\alpha)\,d\pi_\alpha
=(T_g^{*}u_r)(\gamma)=\sigma_r\,v_r(\gamma),  \label{eq:gammacoef}
\end{equation}
a \emph{scalar} coefficient. Hence $\Psi(\gamma)=\sum_r\sigma_r
v_r(\gamma)u_r$, lying in $\overline{\spn}\{u_r\}\subseteq\mathcal{H}%
_{\alpha}$ for every $\gamma$ by the argument of Lemma~\ref{lem:gradexp},
and Parseval gives the $\gamma$-analogue of \eqref{eq:parseval}, 
\begin{equation*}
\lVert\Psi(\gamma_1)-\Psi(\gamma_2)\rVert_{\mathcal{H}_{\alpha}}^2
=\sum_{r\ge1}\sigma_r^2\bigl(v_r(\gamma_1)-v_r(\gamma_2)\bigr)^2 .
\end{equation*}
Differentiating \eqref{eq:gammacoef} in $\gamma$ and applying the entrywise
Bessel and projection argument of Lemma~\ref{lem:gradexp} to $D_\gamma g$
gives the $\gamma$-side gradient information matrix 
\begin{equation*}
N_\infty(\gamma)=\sum_{r\ge1}\sigma_r^2\,Dv_r(\gamma)^{\top}Dv_r(\gamma)
=\int_{\Omega_{\alpha}}D_\gamma g(\alpha,\gamma)^{\top}D_\gamma
g(\alpha,\gamma)\, d\pi_\alpha .
\end{equation*}
Write $V_R(\gamma)=(v_1(\gamma),\dots,v_R(\gamma))^{\top}\in\mathbb{R}^{R}$, 
$N_R(\gamma):=\sum_{r\le R}\sigma_r^2\,Dv_r(\gamma)^{\top}Dv_r(\gamma)$, and 
$Z_\gamma:=\{\gamma\in\Omega_{\gamma}:\det N_\infty(\gamma)=0\}$. The $%
\gamma $-analogue of an Assumption is that Assumption with the roles of $%
\alpha$ and $\gamma$ exchanged. Two are used below. The $\gamma$-analogue of
Assumption~\ref{ass:obseq} asks that $g(\cdot,\gamma_1)=g(\cdot,\gamma_2)$
on $\Omega_{\alpha}$ imply $\gamma_1=\gamma_2$, i.e.\ that $\Psi$ be
injective. The $\gamma$-analogue of Assumption~\ref{ass:multi-index} asks
that there exist $a_0=(\alpha_{0,1},\dots,\alpha_{0,p_\gamma})\in\Omega_{%
\alpha}^{p_\gamma}$ with $d_{\gamma}\le d_G\,p_\gamma$ and $p_\gamma<\infty$
such that for every $\gamma\in\Omega_{\gamma}$ the stacked Jacobian 
\begin{equation*}
D_\gamma G_2(a_0,\gamma) :=\Bigl[\,D_\gamma g(\alpha_{0,1},\gamma)\,;\
\cdots\ ;\, D_\gamma g(\alpha_{0,p_\gamma},\gamma)\,\Bigr] \in\mathbb{R}%
^{d_G p_\gamma\times\,d_{\gamma}}
\end{equation*}
has full column rank $d_{\gamma}$. Since both sides take the same
singular values, $\max(R_0,R_0^{\prime })$ is a truncation order in a single
ordering.

\begin{corollary}[Finite truncation]
\label{cor:gamma} Under the $\gamma$-analogues of Assumptions~\ref%
{ass:domain}, \ref{ass:obseq}, and~\ref{ass:multi-index}, there is a finite $%
R_0^{\prime }$ such that $V_R$ is injective on $\Omega_{\gamma}$ for all $%
R\ge R_0^{\prime }$. If the assumptions of Theorem~\ref{thm:main} hold as
well, taking $R\ge\max(R_0,R_0^{\prime })$ makes $U_R$ and $V_R$
simultaneously injective.
\end{corollary}

\begin{proof}
The proof of Theorem~\ref{thm:main} applies verbatim under the dictionary 
\begin{equation*}
(\Phi,\{u_r\},M_R,M_\infty,Z,\Omega_{\alpha},\pi_\alpha)\;\longrightarrow\;
(\Psi,\{v_r\},N_R,N_\infty,Z_\gamma,\Omega_{\gamma},\pi_\gamma):
\end{equation*}
the Parseval identity displayed above replaces \eqref{eq:parseval}, so
Proposition~\ref{prop:part1} transfers with the $\gamma$-analogue of
Assumption~\ref{ass:obseq} supplying injectivity of $\Psi$ and hence of $%
V_\infty$; the gradient expansion for $D_\gamma g$ noted above replaces
Lemma~\ref{lem:gradexp}; the $\gamma$-analogue of Assumption~\ref%
{ass:multi-index} yields the floor $N_\infty(\gamma)\succeq c_0^{\prime
}I_{d_{\gamma}}$ on $\Omega_{\gamma}$ by the argument of Lemma~\ref%
{lem:Minf-uniform}; and Lemma~\ref{lem:near} together with the Dini argument
of Theorem~\ref{thm:main} use only these objects and the compactness of $%
\Omega_{\gamma}$.
\end{proof}

\begin{corollary}[Almost-everywhere injectivity]
\label{cor:gamma-ae} Under the $\gamma$-analogues of Assumptions~\ref%
{ass:domain}, \ref{ass:obseq}, and~\ref{ass:ae-immersion}, for every compact 
$K\subseteq\Omega_{\gamma}\setminus Z_\gamma$ there is a finite $R_K^{\prime
}$ with $V_R$ injective on $K$ for all $R\ge R_K^{\prime }$. Hence $%
\pi_\gamma$-almost every $\gamma$ has a neighbourhood on which a finite
leading truncation is injective.
\end{corollary}

\begin{proof}
The proof of Proposition~\ref{prop:ae} applies under the same dictionary as
in the proof of Corollary~\ref{cor:gamma}, with the $\gamma$-analogue of
Assumption~\ref{ass:ae-immersion} reading $\pi_\gamma(Z_\gamma)=0$.
\end{proof}

\begin{remark}
\label{rem:transpose} Neither corollary is obtained by applying Theorem~\ref%
{thm:main} or Proposition~\ref{prop:ae} to the transposed function $\tilde
g(\gamma,\alpha):=g(\alpha,\gamma)$. When $d_G=1$ that route is available: $%
T_{\tilde g}=T_g^{*}$, and the singular system of $\tilde g$ is that of $g$
with the roles of $u_r$ and $v_r$ exchanged. When $d_G>1$ it is not: the
operator $T_{\tilde g}:L^{2}_{\pi}(\Omega_{\alpha})\to\mathcal{H}_{\gamma}$, 
$(T_{\tilde
g}\phi)(\gamma)=\int_{\Omega_{\alpha}}g(\alpha,\gamma)\phi(\alpha)\,d\pi_%
\alpha$, carries its own singular system, and its scalar functions on $%
\Omega_{\alpha}$ diagonalise the pooled kernel $\int_{\Omega_{\gamma}}g(%
\alpha_1,\gamma)^{\top}g(\alpha_2,\gamma)\,d\pi_\gamma$; in general this is
a different decomposition from \eqref{eq:svd}, with different singular
values and functions. The mirrored argument above stays within the single
system of \eqref{eq:svd}, whose $\gamma$-side coefficients are taken from
the adjoint \eqref{eq:gammacoef}; this is also what makes $%
\max(R_0,R_0^{\prime })$ meaningful in one ordering.
\end{remark}

\begin{proof}[Proof of Theorem~\protect\ref{thm:headline}]
(i) Corollary~\ref{cor:blend} provides $\Omega_\delta^\alpha$ and a finite
order $R_\alpha(\delta)$. The same construction on the $\gamma$-side, with
the sets $T_\eta$ built around the closed $Z_\gamma$ and Corollary~\ref%
{cor:gamma-ae} in place of Proposition~\ref{prop:ae}, provides $%
\Omega_\delta^\gamma$ and $R_\gamma(\delta)$. Take $R(\delta)=\max\{R_%
\alpha(\delta),R_\gamma(\delta)\}$, well defined in a single ordering since
both sides have the same singular values. (ii) Theorem~\ref{thm:main} gives $%
R_0$ and Corollary~\ref{cor:gamma} gives $R_0^{\prime }$; take $%
R^{*}=\max(R_0,R_0^{\prime })$. Identification of the proxies is Lemma~\ref%
{lem:proxies-identified} and injectivity is (i) and (ii), so $(U_R,V_R)$ are
valid proxies under (ii), and Theorem~\ref{th:proxyID} then identifies $h$,
on the indicated sets under (i).
\end{proof}

\subsection{Additive unobserved heterogeneity}

\label{app:inj-additive}

\begin{proposition}[Additive singular functions]
\label{prop:additive-eigen} Suppose Assumptions~\ref{ass:domain} and~\ref%
{ass:additive} hold. For every $r$ with $\sigma_r>0$ there are continuous $%
\mathbb{R}^{d_G}$-valued functions $u_{r,k}$ on $\Omega_{\alpha,k}$, $C^1$
under Assumption~\ref{ass:domain}, such that 
\begin{equation*}
u_r(\alpha)=\sum_{k=1}^{d}u_{r,k}(\alpha_k) \quad\text{for every }%
\alpha\in\Omega_{\alpha}, \qquad
u_{r,k}(s)=\sigma_r^{-1}\int_{\Omega_{\gamma,k}} g_k(s,c)\,\bar
v_{r,k}(c)\,d\pi_{\gamma,k}(c),
\end{equation*}
and symmetrically $v_r(\gamma)=\sum_{k=1}^{d}v_{r,k}(\gamma_k)$. The
decomposition is unique up to constants summing to zero: if $%
\sum_{k}w_k(\alpha_k)=0$ on $\Omega_{\alpha}$ with each $w_k$ continuous,
then every $w_k$ is constant. In particular the components are identified
from the $u_r$ once normalised, say $u_{r,k}(\alpha^{\circ}_k)=0$ for $k\ge2$%
.
\end{proposition}

\begin{proof}
By Assumption~\ref{ass:additive}, for every $\alpha\in\Omega_{\alpha}$, 
\begin{equation*}
(T_g v_r)(\alpha)
=\int_{\Omega_{\gamma}}g(\alpha,\gamma)\,v_r(\gamma)\,d\pi_\gamma(\gamma)
=\sum_{k=1}^{d}\int_{\Omega_{\gamma}}g_k(\alpha_k,\gamma_k)\,v_r(\gamma)\,
d\pi_\gamma(\gamma),
\end{equation*}
and since the $k$-th integrand depends on $\gamma$ only through its $k$-th
coordinate, conditioning on $\gamma_k$ (iterated expectations) reduces each
term to an integral against the marginal: with $\bar v_{r,k}(c)=\mathbb{E}%
[v_r(\gamma_t)\mid\gamma_{t,k}=c]$, 
\begin{equation*}
\int_{\Omega_{\gamma}}g_k(\alpha_k,\gamma_k)\,v_r(\gamma)\,d\pi_\gamma(%
\gamma) =\int_{\Omega_{\gamma,k}}g_k(\alpha_k,c)\, \bar
v_{r,k}(c)\,d\pi_{\gamma,k}(c),
\end{equation*}
each summand a function of $\alpha_k$ alone. Since $u_r=\sigma_r^{-1}T_g v_r$
in $\mathcal{H}_{\alpha}$ and both sides are continuous while $\pi_\alpha$
has full support (density bounded below), the identity holds at every $%
\alpha $. Continuity of $u_{r,k}$ follows from dominated convergence ($g_k$
is bounded and continuous, $\bar v_{r,k}\in L^1(\pi_{\gamma,k})$), and $C^1$
from differentiation under the integral, $\partial_s g_k$ being bounded and
continuous under Assumption~\ref{ass:domain}. For uniqueness, fix the
reference point $\alpha^{\circ}$ and, given $k$, evaluate $\sum_j
w_j(\alpha_j)=0$ at the point with $j$-th coordinate $\alpha^{\circ}_j$ for $%
j\neq k$ and $k$-th coordinate $s$ free, which the product support permits: $%
w_k(s)=-\sum_{j\neq k}w_j(\alpha^{\circ}_j)$ is constant. The statement for $%
v_r$ follows from the adjoint by the same computation: $(T_g^{*}u_r)(%
\gamma)=\int_{\Omega_{\alpha}}g(\alpha,\gamma)^{\top}u_r(\alpha)\,d\pi_%
\alpha =\sum_k\int_{\Omega_{\alpha,k}}\bar
u_{r,k}(s)^{\top}g_k(s,\gamma_k)\, d\pi_{\alpha,k}(s)$, with $\bar
u_{r,k}(s):=\mathbb{E}[u_r(\alpha_i)\mid\alpha_{i,k}=s]$, each summand a
function of $\gamma_k$ alone; divide by $\sigma_r$.
\end{proof}

\begin{proof}[Proof of Theorem~\protect\ref{thm:additive}]
\emph{(i).} Fix $k$ and $s,s^{\prime }\in\Omega_{\alpha,k}$ with $%
U^{(k)}_R(s)=U^{(k)}_R(s^{\prime })$, and let $\alpha[s]\in\Omega_{\alpha}$
denote the point with $k$-th coordinate $s$ and $j$-th coordinate $%
\alpha^{\circ}_j$ for $j\neq k$, which lies in $\Omega_{\alpha}$ by the
product structure. By Proposition~\ref{prop:additive-eigen}, for every $r\le
R$ the components in coordinates $j\neq k$ cancel in the difference, so 
\begin{equation*}
u_r(\alpha[s])-u_r(\alpha[s^{\prime }]) =u_{r,k}(s)-u_{r,k}(s^{\prime })=0 ,
\end{equation*}
and note this holds for \emph{any} additive decomposition, the constants
cancelling as well. Hence $U_R(\alpha[s])=U_R(\alpha[s^{\prime }])$, and
Theorem~\ref{thm:main} ($R\ge R_0$) gives $\alpha[s]=\alpha[s^{\prime }]$,
i.e.\ $s=s^{\prime }$. The $\gamma$-side is symmetric, via Corollary~\ref%
{cor:gamma}.

\emph{(ii).} The $u_r(\alpha_i)$ are identified by Lemma~\ref%
{lem:proxies-identified}, hence so are $u_{r,k}$ by the uniqueness statement
of Proposition~\ref{prop:additive-eigen} under the stated normalisation.
Injectivity of each block in its own coordinate gives injectivity of the
stacked map: if $\lambda_i=\lambda_i^{\prime }$ then $U^{(k)}_R(%
\alpha_{i,k})=U^{(k)}_R(\alpha_{i,k}^{\prime })$ for every $k$, so $%
\alpha_{i,k}=\alpha_{i,k}^{\prime }$ by (i).

\emph{(iii).} By (ii) the stacked maps $U$ and $V$ are valid proxies in the
sense of Definition~\ref{def:proxyID}, and the set $\{g(\alpha_i,\gamma_t)\}$
is identified by hypothesis. Theorem~\ref{th:proxyID} therefore gives $%
h(\lambda_i,f_t)=g\bigl(U^{-1}(\lambda_i),V^{-1}(f_t)\bigr)%
=g(\alpha_i,\gamma_t) =\sum_k g_k(\alpha_{i,k},\gamma_{t,k})$, the last
equality by Assumption~\ref{ass:additive}. By (i) each $U^{(k)}_R$ is a
continuous injection on a compact set, hence a homeomorphism onto its image
with continuous inverse, and likewise $V^{(k)}_R$; substituting $%
\alpha_{i,k}=(U^{(k)}_R)^{-1}(\lambda_{i,k})$ and $%
\gamma_{t,k}=(V^{(k)}_R)^{-1}(f_{t,k})$ gives the display, with $h_k$
continuous as a composition of continuous maps.
\end{proof}

\section{Primitive smoothness conditions}

\label{sect:sobolev}

For $p\in\mathbb{N}$ and a bounded domain $\Omega$, $H^p(\Omega)$ denotes
the order-$p$ $L^2$ Sobolev space under Lebesgue measure and, for a measure $%
\Pi$ on $\Omega$ with density $\pi$, its weighted counterpart is 
\begin{equation*}
H^p(\Pi;\mathbb{R}^{k}):=\bigl\{u:\Omega\to\mathbb{R}^{k}\ :\ D^s u\in
L^2(\Pi),\ |s|\le p\bigr\}, \qquad \lVert u\rVert_{H^p(\Pi)}^2:=\sum_{|s|\le
p}\int_\Omega\lVert D^s u\rVert^2\,d\Pi ,
\end{equation*}
where $L^2(\Pi)$ is $L^2$ with respect to joint measure $\Pi$. The scalar
case $k=1$ being written $H^p(\Pi)$. The two coincide under the standing
bounds on the densities. In particular \emph{no smoothness of the measures
is required anywhere}, since $\pi_\alpha$ and $\pi_\gamma$ are only ever
integrated against, never differentiated.

\begin{lemma}
\label{lem:weighted} If $0<c\le\pi\le C$ on $\Omega$, then $%
H^p(\Pi)=H^p(\Omega)$ as sets, with 
\begin{equation*}
c\,\lVert u\rVert_{H^p(\Omega)}^2 \;\le\;\lVert u\rVert_{H^p(\Pi)}^2
\;\le\;C\,\lVert u\rVert_{H^p(\Omega)}^2 .
\end{equation*}
This applies to $\pi_\alpha$, $\pi_\gamma$, and the product $%
\pi_\alpha\otimes\pi_\gamma$, whose densities are bounded away from $0$ and $%
\infty$.
\end{lemma}

\begin{proof}
Each term of the norm is $\int_\Omega|D^s u|^2\,\pi\,dx$ with $|s|\le p$,
and the bounds on $\pi$ sandwich it between $c\int|D^s u|^2dx$ and $%
C\int|D^s u|^2dx$. No derivative falls on $\pi$, so no regularity of $\pi$
is used.
\end{proof}

Lemma~\ref{lem:weighted} states that the sets induced by the weighed and
non-weighted Sobolev spaces of singular functions are the same, i.e. $%
H^p(\Pi)=H^p(\Omega)$, and that the weighting also induces equivalent norms
over these spaces. Hence, singular functions admitted by $H^p(\Pi)$, and
related spectral decay, inherit the properties of the unweighted $%
H^p(\Omega) $, which have been studied thoroughly, see e.g. \cite%
{griebel2014approximation,griebel2019singular}. The results relevant to our
setting follow.

Assumption~\ref{ass:domain} requires each $u_r$ to be $C^1$. No separate
condition is needed for this. Since $u_r=\sigma_r^{-1}T_g v_r$ and $%
v_r=\sigma_r^{-1}T_g^{*}u_r$, derivatives fall on $g$ inside the integral,
so the singular functions inherit the smoothness of $g$. The following
theorem, adapted from \citet[Lemmata~3.2
and~3.7]{griebel2014approximation} and %
\citet[Proposition~3.5]{griebel2019singular} to our weighted measures, makes
this primitive; part~(iii) uses the sharp eigenvalue-decay rate of the
latter, which improves Theorem~3.3 of the former by an additive $1/2$ in the
singular-value exponent.\footnote{%
The exponent printed in their Proposition~3.5 reads $\min\{p_1/n_1,p_2/n_2\}$%
; their Corollary~3.4, from which the proposition follows, gives $%
\max\{p_1/n_1,p_2/n_2\}$, i.e.\ $p/\min(d_{\alpha},d_{\gamma})$ here. The
proof below is self-contained and does not rely on the printed statement.}

\begin{theorem}
\label{thm:primitive} Let $\Omega_{\alpha}$ and $\Omega_{\gamma}$ be
closures of bounded Lipschitz domains (for $\Omega_{\alpha}$ this follows
from compactness and convexity once $\Omega_{\alpha}$ has nonempty interior,
which the density bounds imply), and let $g\in
H^p(\Omega_{\alpha}\times\Omega_{\gamma};\mathbb{R}^{d_G})$ for some $p\in%
\mathbb{N}$. Then:

\begin{enumerate}
\item[\textup{(i)}] \textup{(Inheritance.)} For every $r$ with $\sigma_r>0$, 
$u_r\in H^p(\Omega_{\alpha};\mathbb{R}^{d_G})$ and $v_r\in
H^p(\Omega_{\gamma})$, with 
\begin{equation*}
\lVert u_r\rVert_{H^p(\Omega_{\alpha};\mathbb{R}^{d_G})}\lesssim%
\sigma_r^{-1}\lVert g\rVert_{H^p}, \qquad \lVert
v_r\rVert_{H^p(\Omega_{\gamma})}\lesssim\sigma_r^{-1}\lVert g\rVert_{H^p}.
\end{equation*}

\item[\textup{(ii)}] \textup{(Embedding.)} If $p>1+(d_{\alpha}+d_{\gamma})/2$%
, then the continuous version of $g$ is $C^1$ on $\Omega_{\alpha}\times%
\Omega_{\gamma}$ with $D_\alpha g$ and $D_\gamma g$ jointly continuous, so
the smoothness required by Assumption~\ref{ass:domain} holds. If merely $%
p>1+d_{\alpha}/2$, then already each $u_r\in C^1(\Omega_{\alpha})$ by 
\textup{(i)} (symmetrically, $p>1+d_{\gamma}/2$ gives each $v_r\in
C^1(\Omega_{\gamma})$).

\item[\textup{(iii)}] \textup{(Spectral decay.)} The tail sums obey $%
\sum_{r>M}\sigma_r^2\lesssim M^{-2p/\min(d_{\alpha},d_{\gamma})}\,\lVert
g\rVert_{H^p}^2$, and consequently $\sigma_r\lesssim
r^{-p/\min(d_{\alpha},d_{\gamma})-1/2}\,\lVert g\rVert_{H^p}$.
\end{enumerate}
\end{theorem}

\begin{proof}
(i) For a multi-index $s$ with $|s|\le p$, differentiating under the
integral sign in \eqref{eq:Tg} (weak derivatives, Fubini) gives $%
D^s_\alpha(T_g\phi)(\alpha) =\int_{\Omega_{\gamma}}D^s_\alpha
g(\alpha,\gamma)\,\phi(\gamma)\,d\pi_\gamma(\gamma)$, and by Cauchy--Schwarz 
\begin{equation*}
\lVert D^s_\alpha(T_g\phi)\rVert_{L^2(\pi_\alpha)}^2
\le\int_{\Omega_{\alpha}}\Bigl(\int_{\Omega_{\gamma}}\lVert D^s_\alpha
g\rVert^2 d\pi_\gamma\Bigr) d\pi_\alpha\;\lVert
\phi\rVert_{L^{2}_{\pi}(\Omega_{\gamma})}^2 \;\lesssim\;\lVert
g\rVert_{H^p}^2\,\lVert \phi\rVert_{L^{2}_{\pi}(\Omega_{\gamma})}^2,
\end{equation*}
the last step passing between weighted and Lebesgue norms via Lemma~\ref%
{lem:weighted}. Hence $T_g:L^{2}_{\pi}(\Omega_{\gamma})\to
H^p(\Omega_{\alpha};\mathbb{R}^{d_G})$ is bounded with norm $\lesssim\lVert
g\rVert_{H^p}$ \citep[Lemma~3.2]{griebel2014approximation}, and $%
u_r=\sigma_r^{-1}T_g v_r$ with $\lVert
v_r\rVert_{L^{2}_{\pi}(\Omega_{\gamma})}=1$ gives the first bound %
\citep[Lemma~3.7]{griebel2014approximation}; the bound for $v_r$ is
symmetric, via $T_g^{*}$. Note that only the $\alpha$-derivatives of $g$
enter the $u$-side (and only the $\gamma$-derivatives the $v$-side), which
is why mixed-smoothness refinements apply \citep{griebel2019singular}.

(ii) A product of bounded Lipschitz domains is a bounded Lipschitz domain,
and on such a domain of dimension $d$ the Sobolev embedding gives $%
H^p\hookrightarrow C^1$ up to the closure whenever $p>1+d/2$. Apply this on $%
\Omega_{\alpha}\times\Omega_{\gamma}$ (dimension $d_{\alpha}+d_{\gamma}$)
for $g$, and on $\Omega_{\alpha}$ (dimension $d_{\alpha}$) for $u_r$ using 
\textup{(i)}.

(iii) Write $\lambda_r:=\sigma_r^2$. By the best-approximation property of
the truncated SVD in the Hilbert--Schmidt norm (Schmidt's theorem), for 
\emph{any} rank-$M$ function $g_M$, Lemma~\ref{lem:weighted} gives 
\begin{equation*}
\sum_{r>M}\lambda_r \;\le\;\lVert
g-g_M\rVert_{L^2(\pi_\alpha\otimes\pi_\gamma)}^2 \;\lesssim\;\lVert
g-g_M\rVert_{L^2}^2 .
\end{equation*}
Take $g_M=(P_M\otimes I)g$, with $P_M$ the $L^2$-orthogonal projection onto
a space of $M$ discontinuous piecewise polynomials of total degree $\lceil
p\rceil$ on a quasi-uniform triangulation of whichever of $\Omega_{\alpha}$, 
$\Omega_{\gamma}$ has the smaller dimension; this function has rank at most $%
d_G M$, each of the $M$ scalar basis functions contributing at most $d_G$
directions to $\text{range}T_{g_M}$, so the tail bound is obtained at index $%
d_G M$ and re-indexed at the cost of a fixed factor. The Bramble--Hilbert
bound gives $\lVert g-g_M\rVert_{L^2}\lesssim
M^{-p/\min(d_{\alpha},d_{\gamma})}\lVert g\rVert_{H^p}$ 
\citep[eq.~(3.7) and
Theorem~3.2]{griebel2019singular}. This proves the tail bound. Monotonicity
of $\{\lambda_r\}$ then converts tails into the pointwise rate: $%
M\lambda_{2M}\le\sum_{r=M+1}^{2M}\lambda_r\lesssim
M^{-2p/\min(d_{\alpha},d_{\gamma})}\lVert g\rVert_{H^p}^2$, so $%
\lambda_r\lesssim r^{-2p/\min(d_{\alpha},d_{\gamma})-1}\lVert
g\rVert_{H^p}^2 $, which is the claim. This is the weighted transfer of %
\citet[Proposition~3.5]{griebel2019singular}, which sharpens %
\citet[Theorem~3.3]{griebel2014approximation} by the additive $1$ in the
eigenvalue exponent and is itself sharp; the weights enter only through the
constant of Lemma~\ref{lem:weighted}.
\end{proof}

\begin{remark}
\label{rem:ambient} Theorem~\ref{thm:primitive} delivers smooth singular
functions without changing the underlying Hilbert spaces. The smoothness is
inherited from $g$, not imposed through the ambient space. This matters
because re-ambienting the SVD in a Sobolev space would change the adjoint $%
T_g^{*}$ and hence the decomposition itself with different singular values
and functions. In Proposition~\ref{prop:part1} (Parseval), the separation
function \eqref{eq:sep}, and the closed form \eqref{eq:Minf} all rely on $%
L^2 $-orthonormality of $\{v_r\}$; the $L^2$-SVD is also the population
counterpart of the sample singular value decomposition used in estimation.
Part~\textup{(iii)} quantifies, in addition, the singular-value decay.
\end{remark}

\section{Analytic sufficiency for the almost-everywhere immersion}

\label{sect:analytic} 

Assumption~\ref{ass:ae-immersion} restricts the degeneracy set $%
Z=\{\alpha\in\Omega_{\alpha}:\det M_\infty(\alpha)=0\}$ to a $\pi_\alpha$%
-null set. This appendix gives a primitive sufficient condition. Real
analyticity of $g$ in $\alpha$, together with the observational-equivalence
normalisation of Assumption~\ref{ass:obseq}, already delivers the null
degeneracy set, so no separate rank condition need be imposed in the
analytic class. The proof also isolates the division of labour. Assumption~%
\ref{ass:obseq} is what rules out degeneracy \emph{everywhere}, and that
step needs only Assumption~\ref{ass:domain}, while analyticity is what
upgrades ``not everywhere'' to ``almost nowhere''. Remark~\ref%
{rem:analytic-sharp} shows the upgrade genuinely fails in the $C^\infty$
class.

\begin{lemma}[Analytic unobserved heterogeneity has a null degeneracy set]
\label{lem:analytic} Let Assumptions~\ref{ass:domain} and~\ref{ass:obseq}
hold, let $\Omega_{\alpha}$ have nonempty interior, and let $\pi_\alpha$ be
absolutely continuous with respect to Lebesgue measure on $\mathbb{R}%
^{d_{\alpha}}$. Suppose in addition that $g$ is real analytic in $\alpha$,
in the sense that there are an open set $U\subseteq\mathbb{C}^{d_{\alpha}}$
containing $\Omega_{\alpha}$ and an extension $\tilde
g:U\times\Omega_{\gamma}\to\mathbb{C}^{d_G}$ of $g$ with $\tilde
g(\cdot,\gamma)$ holomorphic on $U$ for every $\gamma\in\Omega_{\gamma}$ and 
$\sup_{\gamma\in\Omega_{\gamma}}\sup_{\alpha\in K}\lVert \tilde
g(\alpha,\gamma)\rVert<\infty$ for every compact $K\subset U$. Then $\det
M_\infty$ is real analytic on the interior of $\Omega_{\alpha}$ and not
identically zero on $\Omega_{\alpha}$, the set $Z$ is Lebesgue-null, and
consequently $\pi_\alpha(Z)=0$, so Assumption~\ref{ass:ae-immersion} holds.
\end{lemma}

\noindent The extension hypothesis is implied by joint real analyticity of $%
g $ on an open neighbourhood of $\Omega_{\alpha}\times\Omega_{\gamma}$ by a
compactness argument, and holds trivially when $g$ is a polynomial or an
entire function of $\alpha$.

\begin{proof}
\emph{Step 1: $M_\infty$ is real analytic.} Each entry of the Jacobian $%
\tilde J(\alpha,\gamma):=D_\alpha\tilde g(\alpha,\gamma)$ is holomorphic in $%
\alpha$ on $U$ for fixed $\gamma$, and the Cauchy integral formula for
derivatives converts the locally uniform bound on $\tilde g$ into a bound $%
\sup_{\gamma\in\Omega_{\gamma}}\sup_{\alpha\in K^{\prime }}\lVert \tilde
J(\alpha,\gamma)\rVert<\infty$ on every compact $K^{\prime }\subset U$. On
real $\alpha$ the extension restricts to the Jacobian $J$ of Assumption~\ref%
{ass:domain}. Fix indices $j,k$ and set $m_{jk}(\alpha):=\int_{\Omega_{%
\gamma}}\bigl(\tilde J(\alpha,\gamma)^{\top} \tilde J(\alpha,\gamma)\bigr)%
_{jk}\,d\pi_\gamma(\gamma)$ for $\alpha\in U$. The integrand is holomorphic
in $\alpha$, measurable in $\gamma$, and locally uniformly bounded, so $%
m_{jk}$ is holomorphic on $U$ (Morera's theorem with Fubini, or dominated
convergence applied to difference quotients). By the closed form %
\eqref{eq:Minf}, the restriction of $m_{jk}$ to real $\alpha$ is the $(j,k)$
entry of $M_\infty$, which is therefore real analytic on $\text{int}%
\Omega_{\alpha}$, and $\det M_\infty$, a polynomial in the entries, is real
analytic there as well.

\emph{Step 2: a dichotomy.} Suppose $Z$ has positive Lebesgue measure. The
boundary of the compact convex $\Omega_{\alpha}$ is Lebesgue-null, so $Z\cap%
\text{int}\Omega_{\alpha}$ has positive measure. The set $\text{int}%
\Omega_{\alpha}$ is open, nonempty by hypothesis, and connected by
convexity, and a real analytic function on a connected open set that
vanishes on a set of positive Lebesgue measure vanishes identically %
\citep{krantz2002primer,mityagin2020zero}. Hence $\det M_\infty\equiv0$ on $%
\text{int}\Omega_{\alpha}$, and by continuity (Lemma~\ref{lem:Minf-uniform})
on all of $\Omega_{\alpha}$. It therefore suffices to show that identical
vanishing of $\det M_\infty$ contradicts Assumption~\ref{ass:obseq}.

\emph{Step 3: identical degeneracy contradicts observational equivalence.}
Suppose $\det M_\infty(\alpha)=0$ for every $\alpha\in\Omega_{\alpha}$.
Three observations combine.

First, a null vector of $M_\infty(\alpha)$ annihilates the whole section.
For a unit vector $v\in\mathbb{R}^{d_{\alpha}}$, the closed form %
\eqref{eq:Minf} gives $v^{\top}M_\infty(\alpha)v=\int_{\Omega_{\gamma}}%
\lVert J(\alpha,\gamma)v\rVert^2d\pi_\gamma$, so $M_\infty(\alpha)v=0$
forces $J(\alpha,\gamma)v=0$ for $\pi_\gamma$-almost every $\gamma$ and
then, by continuity of $J$ in $\gamma$ and the support convention $%
\Omega_{\gamma}=\text{supp}\pi_\gamma$, for every $\gamma\in\Omega_{\gamma}$%
, as in the proof of Lemma~\ref{lem:Minf-uniform}.

Second, the null direction admits a continuous local selection. Let $%
\rho^{*} $ denote the maximum of $\rank M_\infty(\alpha)$ over $\alpha\in%
\text{int}\Omega_{\alpha}$, so that $\rho^{*}\le\,d_{\alpha}-1$. The
maximal-rank stratum $S:=\{\alpha\in\text{int}\Omega_{\alpha}:\rank %
M_\infty(\alpha)=\rho^{*}\}$ is nonempty and open, since the rank is lower
semicontinuous and cannot exceed $\rho^{*}$. Fix $\alpha^{*}\in S$. On a
small ball $B\subseteq S$ around $\alpha^{*}$ the $\rho^{*}$ nonzero
eigenvalues of the symmetric positive-semidefinite $M_\infty(\alpha)$ remain
above some $c>0$ while the remaining $d_{\alpha}-\rho^{*}\ge1$ eigenvalues
are exactly zero, so the Riesz projection 
\begin{equation*}
P(\alpha)=\frac{1}{2\pi i}\oint_{|z|=c/2}\bigl(zI-M_\infty(\alpha)\bigr)%
^{-1}dz
\end{equation*}
is the orthogonal projector onto $\ker M_\infty(\alpha)$, and it is
continuous in $\alpha$ because $M_\infty$ is continuous and the resolvent
depends continuously on the matrix, uniformly over the compact contour. For
a unit $v_0\in\ker M_\infty(\alpha^{*})$, the field $v(\alpha):=P(%
\alpha)v_0/\lVert P(\alpha)v_0\rVert$ is well defined and continuous on a
smaller ball where $\lVert P(\alpha)v_0\rVert\ge1/2$, with $v(\alpha)$ a
unit vector in $\ker M_\infty(\alpha)$.

Third, flowing along the selection produces observationally equivalent
points. By Peano's existence theorem there is a $C^1$ curve $\alpha(t)$, $%
t\in(-\varepsilon,\varepsilon)$, with $\alpha(0)=\alpha^{*}$, $%
\dot\alpha(t)=v(\alpha(t))$, and $\alpha(t)\in B$ throughout. For every
fixed $\gamma\in\Omega_{\gamma}$, 
\begin{equation*}
\frac{d}{dt}\,g\bigl(\alpha(t),\gamma\bigr) =J\bigl(\alpha(t),\gamma\bigr)\,v%
\bigl(\alpha(t)\bigr)=0
\end{equation*}
by the first observation, so $g(\alpha(t),\cdot)=g(\alpha^{*},\cdot)$ on $%
\Omega_{\gamma}$ for every $t$. The curve is nonconstant. Shrinking $%
\varepsilon$ so that $\langle v(\alpha(s)),v(\alpha^{*})\rangle\ge1/2$ for $%
|s|<\varepsilon$, which continuity permits, gives $\langle\alpha(t)-%
\alpha^{*},v(\alpha^{*})\rangle =\int_0^t\langle
v(\alpha(s)),v(\alpha^{*})\rangle\,ds\ge t/2>0$ for $t\in(0,\varepsilon)$.
Thus $\alpha(t)\ne\alpha^{*}$ while $\Phi(\alpha(t))=\Phi(\alpha^{*})$,
contradicting Assumption~\ref{ass:obseq}. Hence $\det M_\infty$ is not
identically zero, and Step~2 completes the proof.
\end{proof}

\begin{remark}[Scope and sharpness]
\label{rem:analytic-sharp} Analyticity enters the argument only through the
identity theorem of Step~2. Step~3 uses Assumptions~\ref{ass:domain} and~\ref%
{ass:obseq} alone, so under those two assumptions $\det M_\infty$ can never
vanish identically. The upgrade from ``not identically zero'' to ``null zero
set'' is exactly what the analytic class buys, and it fails in the $C^\infty$
class. Let $C\subset[0,1]$ be a fat Cantor set, closed, nowhere dense and of
positive Lebesgue measure, let $\varphi\ge0$ be a $C^\infty$ function
vanishing exactly on $C$, and set $\psi(\alpha)=\int_0^{\alpha}\varphi$,
which is $C^\infty$ and strictly increasing because the complement of $C$ is
open and dense. With $d_{\alpha}=1$, $g(\alpha,\gamma)=\psi(\alpha)\ell(%
\gamma)$ for a nonzero continuous $\ell$, and $\pi_\alpha$ uniform on $[0,1]$%
, Assumption~\ref{ass:obseq} holds, yet $\det
M_\infty(\alpha)=c\,\psi^{\prime 2}$ with $c=\int_{\Omega_{\gamma}}\lVert
\ell(\gamma)\rVert^2d\pi_\gamma>0$, so $Z=C$, $\pi_\alpha(Z)>0$ and
Assumption~\ref{ass:ae-immersion} fails. Finally, the lemma concludes that $%
Z $ is null, not empty. The function $g(\alpha,\gamma)=\alpha^3\ell(\gamma)$
of Remark~\ref{rem:independence} is real analytic with $Z=\{0\}$, so the
uniform Assumption~\ref{ass:multi-index} remains genuinely stronger than
what analyticity supplies.
\end{remark}

\setlength{\bibsep}{2pt} 
\bibliographystyle{chicago3}
\bibliography{refs}

\end{document}